\documentclass[a4paper,11pt]{article}
\pdfoutput=1 

\usepackage{jheppub} 

\usepackage[T1]{fontenc} 

\usepackage{mathtools}  
\newtheorem{lemma}{Lemma}

\newtheorem{theorem}[lemma]{Theorem}
\newtheorem{proof}[lemma]{Proof}
\usepackage{xcolor}

\title{Existence of a Supersymmetric Massless Ground State of the $SU(N)$ Matrix Model globally on its Valleys }
\author{Lyonell Boulton$^1$}
\affiliation{$^1$Maxwell Institute for Mathematical Sciences and Department of Mathematics Heriot-Watt University, Edinburgh, EH14 4AS, United Kingdom.}
\emailAdd{l.boulton@hw.ac.uk}
\author{Mar{\'\i}a Pilar Garc{\'\i}a del Moral$^2$}
\affiliation{$^{2,3}$Departamento de F\'{\i}sica, Universidad de Antofagasta, Aptdo 02800, Chile.}
\emailAdd{maria.garciadelmoral@uantof.cl, alvaro.restuccia@uantof.cl}
\author{Alvaro Restuccia$^3$}

\date{06-01-2021}

\abstract{
In this work we consider the existence and uniqueness of the ground state of the regularized Hamiltonian of the Supermembrane in dimensions $D= 4,\,5,\,7$ and 11, or equivalently the $SU(N)$ Matrix Model. That is, the 0+1 reduction of the 10-dimensional $SU(N)$ Super Yang-Mills Hamiltonian. This ground state problem is associated with the solutions of the inner and outer Dirichlet problems for this operator, and their subsequent smooth patching (glueing)  into a single state. We have discussed properties of the inner problem in a previous work, therefore we now investigate the outer Dirichlet problem for the Hamiltonian operator. We establish existence and uniqueness  on unbounded valleys defined in terms of the bosonic potential. These are precisely those regions where the bosonic part of the potential is less than a given value $V_0$, which we set to be arbitrary. The problem is well posed, since these valleys are preserved by the action of the $SU(N)$ constraint. We first show that their Lebesgue measure is finite, subject to restrictions on $D$ in terms of $N$. We then use this analysis to determine a bound on the fermionic potential which yields the coercive property of the energy form. It is from this, that we derive the existence and uniqueness of the solution. As a by-product of our argumentation, we show that the Hamiltonian, restricted to the valleys, has spectrum purely discrete with finite multiplicity. Remarkably, this is in contrast to the case of the unrestricted space, where it is well known that the spectrum comprises a continuous segment. We discuss the relation of our work with the general ground state problem and the question of confinement in models with strong interactions.
}

\begin{document}

\maketitle
\flushbottom
\section{Introduction} \label{1}
The present paper is devoted to the ground state of the supersymmetric Hamiltonian related to three theories: the regularization of the $D=11$ Supermembrane Theory \cite{dwhn}, the BFSS Matrix Model \cite{bfss} and  the reduction of $D=10$ Super Yang-Mills to $0+1$ dimensions \cite{claudson}. The existence of a massless ground state in any of these three instances  is an open problem.

 The relevance of this regime for (Super)Yang-Mills theories in the infra-red (IR) limit, the so-called slow-mode regime, was highlighted in \cite{gabadadze}. In this IR limit, glueball bound states and flux tubes between quarks are expected to be formed. It has been suggested that they  can be described in terms of confining strings which corresponds to a Nambu-Goto string connecting pairs of quarks at the extremes, subject to several corrections \cite{PS, Polyakov, LW}. See for example \cite{aharony, Brandt, Solberg} or  \cite{Teper, Gliozzi} in the context of lattice QCD. In these works, the confining  strings acquire a width and they have also been modelled out in terms of D2-brane bound states \cite{supertubes}. (Super) membranes are strongly coupled  $2+1$ dimensional objects and their regularized description corresponds to the Matrix Models mentioned above.  It has also been suggested that membranes can be seen as the  IR limit  of Yang Mills theories \cite{Lechtenfeld}  and they have also been used in the literature to describe some aspects of QCD \cite{Ansoldi}. Indeed, in  \cite{inertia} it was shown that the spectrum of the bosonic regularized membrane theory has a mass gap given by the inertia moment and as it is well known the spectrum in this regime is purely discrete for arbitrary $N$ (in particular for $N=3$). The existence of bound states of these theories, realised as eigenvalues embedded in the continuous spectrum, could be an interesting research direction.

Our present goal is to describe the theories, restricted to certain $8\times(N-1)$-dimensional regions determined by the flat directions of the potential. Adhering to the standard terminology, we will call ``valleys'' of the bosonic potential, those points of the space such that the potential is bounded above by a given positive constant $V_0$. That is, they are determined by imposing a given constraint on the ``height'' of the bosonic potential term.  These valleys extend to infinity with decreasing width. Constrained along these valleys, we obtain a $10D$  Super Yang Mills theory in the slow mode regime or the regularized $SU(N)$ supermembrane. \emph{I.e.} the $SU(N)$ Matrix Model confined to a star-shape tubular region. The corresponding Hamiltonian operator has a domain determined by the space of wave functions supported on the valleys, with vanishing boundary condition. Our analysis covers arbitrary rank of the $SU(N)$ gauge group, in particular for $N=2$ and $N=3$, both cases being of interest. 

The Supermembrane Theory was developed in \cite{bst}. The corresponding $SU(N)$ regularization was introduced in \cite{hoppe} and in \cite{dwhn,dwmn} the $SU(N)$ regularized Hamiltonian in the light cone gauge was obtained. The zero mode eigenfunction can be described in terms of the $D=11$ supergravity multiplet, however, the existence of the ground state of the Hamiltonian requires a proof of existence of a unique nontrivial eigenfunction for the nonzero modes. To the best of our knowledge, no complete proof of this fact has been found to this date. Moreover, in order to be identified with the $D=11$ supergravity multiplet, it has to satisfy the additional constraint of being invariant under $SO(9)$.  The existence and uniqueness of the ground state has also been analyzed from different perspectives. One of these started with \cite{dwhn}. Although the problem remains open several interesting contributions to it have been obtained  \cite{dwhn,hoppe,hasler,fh,michishita,hl,hlt,frolich}. We follow this perspective and prove, for a well-defined region around the valleys of the potential extended to infinity, the existence and uniqueness of the nontrivial state annihilating the Hamiltonian. 
Another approach to the problem has been to consider the Witten index. This index is well established for elliptical operators of the Fredholm type in Supersymmetric Quantum Mechanics.  This is not the case with the Hamiltonian \cite{dwhn,bfss}. In \cite{Sethi-Stern01,Sethi-Stern02} an extension of the Witten index for non-Fredholm operators was introduced. Although there are still unsolved problems in the approach, it has been claimed by the authors the existence of a ground state for the $SU(2)$ model and some other extensions. In \cite{Staudacher} contradictions of the approach for exceptional groups were reported. The Witten index was also used in the analysis performed in \cite{Yi}. The index approach, while interesting, does not characterize the ground state wave function beyond belonging to a Hilbert $L_2$ space. In our work, as we said, we follow a completely different approach in the hope of obtaining the existence and uniqueness of the ground state together with bounds for the wave function characterizing the ground state behavior.

In $D=11$ Supermembrane Theory, the zero modes associated with the center of mass and the non-zero modes associated with the internal excitations, decouple. The ground state of the Hamiltonian with zero eigenvalue can be described in terms of  the $D=11$ supergravity multiplet, once the existence of the non-zero modes of a unique nontrivial eigenfunction (with zero eigenvalue) invariant under the R-symmetry $SO(9)$ is proven, \cite{dwhn}. The $SU(N)$ regularized Hamiltonian for nonzero modes coincides with the Hamiltonian of the BFSS Matrix Model, \cite{bfss}. This Hamiltonian was first obtained as the $0+1$ reduction of $10D$ Super Yang Mills \cite{claudson,halpern}. 

In order to find the ground state of the Hamiltonian, we propose three main steps, already sketched in \cite{bgmrGS,bgmrExt,bgmrSU2, bgmrO}. Firstly, determine the existence and uniqueness of the solution to the Dirichlet problem on a bounded region $\Omega$ with smooth boundary $\partial\Omega$. Secondly, determine the existence and uniqueness of the solution to the Dirichlet problem  on the complementary unbounded region. Thirdly, establish the smooth patching along $\partial \Omega$ of both these solutions into a single state which, by construction, is the full ground state. 

We developed the first step in \cite{bgmrGS} for $\Omega$ of arbitrary diameter. Our argument relied on the polynomial form of the bosonic and fermionic potentials, as well as on the supersymmetric structure of the Hamiltonian. In the full space, although the Hamiltonian is positive, the potential becomes negative and arbitrarily large in modulus along certain directions inside the valleys extending to infinity. However, on bounded regions, the potential is bounded (both above and below). Therefore, the Dirichlet form associated to the Hamiltonian restricted to $\Omega$ is coercive. Moreover, the supersymmetric structure of the Hamiltonian together with other analytic properties of potential, imply that a state $\varphi$ constrained to cancellation by the supersymmetric charges in $\Omega$ and satisfying the homogeneous Dirichlet condition $\varphi=0$ on $\partial\Omega$, can only be the null state of the Hilbert space. In turns, the existence and uniqueness of the Dirichlet problem on $\Omega$, follows by standard arguments from the theory of elliptic operators. These involve the Rellich-Kondrashov Compact Embedding Theorem, the Lax-Milgram Theorem and the Fredholm Theorem. See  \cite{bgmrGS} for more specific details. 

The Rellich-Kondrashov Compact Embedding Theorem is valid for every bounded region of the Euclidean space, but unfortunately might fail on unbounded regions. Nonetheless, according to arguments in \cite{Lundholm}, the convergence of the partition function of $0+1$ Yang-Mills \cite{austing1,austing2} is related to the fact that the Lebesgue measure of the bosonic valleys is finite. Following the former work, this property is valid for all algebras considered in \cite{austing1} and \cite{austing2}.

By pursuing an alternative approach to that of  \cite{Lundholm,austing1,austing2}, we established in \cite{bgmrGS} a concrete estimate for the Lebesgue measure of these bosonic valleys for the $SU(2)$ algebra. We then showed that the embedding of $H^1$ into $L^2$ is compact. Hence the Rellich-Kondrashov Theorem is, once again, valid for these regions. Our argumentation was intrinsic to the specific structure of the bosonic potential. From the estimate and a relevant bound for the fermionic potential, it should follow that the Dirichlet form of the complete Hamiltonian is coercive.

One of our main purposes below will be to extend this idea onto an $SU(N)$ algebra and prove the existence and uniqueness of the state, in the physical subspace, annihilated by the Hamiltonian of the $SU(N)$ $D=11$ Supermembrane, restricted to the valleys. Our construction is based on specific properties of the potential, one of the most important being an ellipsoidal symmetry along axial directions. The results, concerning the finiteness of the measure of the valleys, for any $N\ge 2$  agree precisely with the ones already reported in \cite{Lundholm,austing1,austing2}. Additionally, for the $SU(N)$ model we derive an explicit bound for the fermionic potential, proposed in [29] for the $SU(2)$ algebra, which characterizes the wavefunctions of the corresponding Hilbert space. They depend on other properties of the potential, outside a neighbourhood of the origin. We then show that the Dirichlet form associated to the $SU(N)$ Hamiltonian  of the Supermembrane is coercive. From this, and the compact embedding $H^1\subset L^2$ on the valleys, it follows that  the solution of the Dirichlet problem exists and is unique. An important result that follows from our analysis is that the Hamiltonian restricted to the valleys has discrete spectrum with finite multiplicity.

Formally, the potential is dominated by its bosonic component in the directions ``away'' from the valleys, so the wave function is confined along those directions. Although a rigorous proof of the latter for any curve reaching infinity is not currently available, it is clear that the analysis of the Dirichlet problem  along the valleys that we currently conduct, sheds an important light on the direction to follow in the second step of the program mentioned above. 

The remaining of the paper is structured as follows. In Section~\ref{2} we set the scenary by recalling the $0+1$ Matrix Model formulation of the $11D$ Supermembrane Theory in the Light Cone Gauge. In Section~\ref{3} we summarise the bounds we found on the measure of the potential valleys associated to the different $su(N)$ models. The full details of our derivations can be found in Appendices~\ref{A}-\ref{D}. In them, we determine the measure of the valleys when all the eigenvalues of a given configuration $su(N)$ matrix are different. We include detailed discussions of the cases $su(2)$, $su(3)$, $su(4)$ and general $su(N)$ with arbitrary $N$. In Appendix~\ref{C} we show that the Lebesgue measure in all the cases is finite. This implies that only certain dimensions for a given rank of the $SU(N)$ bosonic potential are allowed. In Section~\ref{4} we obtain the associated bounds for the fermionic potential. Section~\ref{5} is devoted to the main results. We establish that the spectrum of the Hamiltonians is purely discrete and demonstrate the existence and uniqueness of the ground state of the theory. In Section~\ref{6} we present a discussion of the ideas presented and our conclusions.  The gauge transformations that we employed in order to show that the valleys have finite measure, are displayed in the final Appendix~\ref{D}.

\section{The \texorpdfstring{$SU(N)$}{}  Matrix Model} \label{2}
We begin by discussing the $0+1$ Matrix Model formulation of the  $11D$ Supermembrane Theory  in the Light Cone Gauge. The latter also corresponds to a $0+1$ reduction of $10D$ Super Yang Mills, known in the literature as the BFSS Matrix Model \cite{bfss}. It describes $D0$-branes interaction among themselves.

 The $D=11$ supermembrane is described in terms of the membrane coordinates $X^m$  and fermionic coordinates $\theta _{\alpha}$,  transforming as a  Majorana spinor on the target space. Both fields are scalar under worldvolume transformations. When the theory is formulated in the Light Cone Gauge the residual symmetries are the global supersymmetry, the R-symmetry $SO(9)$ and a gauge symmetry, the area preserving diffeomorphisms on the base manifold.

Once the theory is regularized by means of the group $SU(N)$, the field operators are labeled by an $SU(N)$ index $A$ and they transform in the adjoint representation of the group. The realization of the wavefunctions is formulated in terms of the $2^{8(N^2-1)}$ irreducible representation of the Clifford algebra. The Hilbert space of physical states consists of the wavefunctions which take values in the fermion Fock space, subject to the $SU(N)$ constraint given by the generator of the $SU(N)$ invariance.

Once it is shown that the zero mode states transform under $SO(9)$ as a $[(44\oplus 84)_{\mathrm{bos}}\oplus 128_{\mathrm{fer}}]$ representation which corresponds to the massless $D=11$ supergravity supermultiplet,
the construction of the ground state wave function reduces to finding a nontrivial solution to
\begin{equation*}\label{A}H\Psi=0\end{equation*}
where $H=\frac{1}{2}M^2$ and $\Psi$ subject to the $SU(N)$ constraint. The latter  is required to be a singlet under $SO(9)$ and $M$ is the mass operator of the supermembrane.
The Hamiltonian associated to the the regularized mass operator of the supermembrane \cite{dwhn} is
\begin{equation}
\begin{aligned}
H&=\frac{1}{2}M^2=-\Delta+V_{\mathrm{B}}+V_F\\ \text{where }&\begin{cases}
\Delta=\frac{1}{2}\frac{\partial^2}{\partial X^i_A\partial X_i^A}\\
V_{B}=\frac{1}{4}f_{AB}^Ef_{CDE}X_i^AX_j^BX^{iC}X^{jD}\\
V_F=if_{ABC}X_i^A\lambda_{\alpha}^B\Gamma_{\alpha\beta}^i\frac{\partial}{\partial\theta_{\beta C}}. \end{cases}
\end{aligned}
\end{equation}
The generators of the local $SU(N)$ symmetry are
\begin{equation}\varphi^A =f^{ABC}\left( X_{i}^B\partial_{X_i^C}
+\theta_{\alpha}^{B}\partial_{\theta_{\alpha}^C}\right).\end{equation}

From the supersymmetric algebra, it follows that the Hamiltonian can be express in terms of the supercharges as
\begin{equation}H =\{Q_{\alpha},Q^{\dagger}_{\alpha}\}\end{equation}
for the physical subspace of solutions, given by the kernel of the first class constraint $\varphi^A$ of the theory. That is
\begin{equation*} \label{constraint} \varphi^A\Psi=0.\end{equation*}

The Hamiltonian $H$ is a positive operator which annihilates $\Psi$, on the physical subspace, if and only if  $\Psi$ is a {singlet} under supersymmetry\footnote{$\Psi_0$, the zero mode wave function, in distinction is a supermultiplet under supersymmetry.}. In such a case,
\begin{equation}Q_{\alpha}\Psi=0\quad \text{and}\quad Q_{\alpha}^{\dagger}\Psi=0.\end{equation} This result does not hold when the theory is restricted by boundary conditions, the case that we will analyze below.

All this ensures that the wavefunction is massless, however it does not guarantee that the ground-state wave function is the corresponding supermultiplet associated to supergravity. For this, $\Psi$ must also become a singlet under $SO(9)$. The spectrum of $H$ in $L^2(\mathbb{R}^n)$ is continuous \cite{dwln}, comprising the segment $[0,\infty)$.

 The bosonic potential can be recast as
\begin{equation}
V_{\mathrm{B}}(X)=-\frac{1}2{\sum_{M,N\ge 1}^dTr[X^M,X^N]^2}=\frac{1}2{\sum_{M,N\ge 1}^dTr[X^M,X^N][X^M,X^N]^{\dag}}\end{equation}
where $X^m= X^{mA}T_A$, considering $X^{mA}$ real coordinates and $T_A$ the generators of the algebra $su(N)$\footnote{The index $A$ corresponds to a pair of indices $(a_1,a_2)$ with $a_i=0,\,\dots,\,\dim(su(N)),\quad i=1,2$ in which the value $(0,0)$ associated to the supermembrane center of mass has being excluded.}; $d=D-2$ corresponds to the number of transverse dimensions of the supermembrane in the LCG. The basis of the $su(N)$ generators satisfies
$[T_A,T_{\mathrm{B}}]=if_{ABC}T_C$ with $Tr (T_AT_{\mathrm{B}})=\delta_{AB}$.

\section{Lebesgue measure of the valleys} \label{3}
Prescribe a height $V_0$. Let
\begin{equation}K\equiv\{X^{mA}: V_{\mathrm{B}}(X)<V_0\}.\end{equation}
We now quote the range of parameters for which the Lebesgue measure of $K$, denoted as $\mathrm{Vol}(K)$, is finite. In the appendices~\ref{A}-\ref{D}, we give precise details of how to derive these conclusions. 

Our argumentation depends on the following simple observation, which we use freely and unambiguously throughout the text.
Let $\widetilde{V}_{\mathrm{B}}(X)$ be another potential expression, such that 
\begin{equation}\widetilde{V}_{\mathrm{B}}(X)\le V_{\mathrm{B}}(X) \qquad \text{ for all matrices }X.\end{equation}
Denote by \begin{equation}\widetilde{K}=\{X^{mA}: \widetilde{V}_{\mathrm{B}}<V_0\}.\end{equation} 
If $\mathrm{Vol}(\widetilde{K})$ is finite, then so is $\mathrm{Vol}(K)$.

In appendices \ref{B4} and \ref{D} we determine, after the evaluation of several estimates, that a restriction on $d$ for each $N$  so that the Lebesgue measure of $K$ is finite turns out to be
\begin{equation}
d>2+\frac{2(N-1)}{3(N-1)-2}.
\end{equation}
This immediately renders the following.
\begin{lemma} \label{lemfinitevolK}
If 
\begin{equation}
d\ge\begin{cases}
 5 & \textrm{and} \quad N=2\\
 4 & \textrm{and} \quad N=3\\
 3 & \textrm{and} \quad N\ge 4,\end{cases}
\end{equation}

then $\operatorname{Vol}(K)<\infty$.
\end{lemma}
In the evaluations leading to the aforementioned observation, we crucially make use of the ellipsoidal symmetry of the bosonic potential. The final result is in agreement with the previous estimates reported in \cite{Lundholm,austing1}. In the next section we will invoke some of our explicit estimates from the appendices, in order to determine a precise bound on the fermionic  potential. This will then be fundamental for our main results, reported in Section~\ref{5}.

\section{Bounds for the fermionic potential} \label{4}

The main point of this section will be to determine an explicit estimate for the fermionic potential on any state. From this estimate, we will show in the next section that the Hamiltonian operator of the Supermembrane Theory, in the admissible dimensions given by Lemma~\ref{lemfinitevolK}, is coercive in the Fock space on $K$. To this end let us recall the properties of the valleys discussed in the appendices. The region we defined as the valleys is a star-shaped region containing directions extending to infinity along which the bosonic potential is zero. On those directions the diagonal components of the matrices tend to infinity while the non-diagonal ones remain bounded. In the appendices we prove that given a large enough distance from the origin there always exists a $SU(N)$ transformation such that the non-diagonal components not only are bounded but decrease as the inverse of a diagonal component, when this one goes to infinity.
We require an explicit expression for $V_{\mathrm{B}}(X)$, and for that we introduce the next convenient notation. The components of a diagonal matrix $\widehat{X}$ are denoted by

\begin{equation}ia_1,\,\dots,\,ia_N \qquad \text{where}\qquad \sum_{i=1}^N a_i=0.\end{equation} 
The other matrices $X^n$, $n=1,\dots,d-1$, have diagonal components

\begin{equation}ib_1^n,\,\dots,\,ib_N^n \qquad \text{where} \qquad \sum_{i=1}^N b_i^n=0\end{equation} and upper-non diagonal components \begin{equation}z_{ij}^n, \quad i<j \qquad \text{where} \qquad i,j=1,\dots,N.\end{equation} Here $a_i$ and $b_i^n$ are real numbers while $z_{ij}^n$ are complex numbers. As the $X^n$ are anti-hermitean, $z_{ji}^n=-\overline{z}_{ij}^n$ where $\overline{z}$ denotes complex conjugation for $j>i$.  We also introduce the vectors $b_i$ with components $b_i^n$, and $z_{ij}$ with components $z_{ij}^n$, $n=1,\dots,d-1$  with norm 

\begin{equation}\vert\vert b_i\vert\vert^2= \sum_{n=1}^{d-1}(b_{i}^{n})^2, \quad \vert\vert z_{ij}\vert\vert^2=\sum_{n=1}^{d-1}z_{ij}^n\overline{z_{ij}^n},\end{equation}
respectively.
The products are defined as $(b_i\cdot z_{jk})=\sum_{n=1}^{d-1}b_i^n z_{jk}^n$ and $(z_{ij} \cdot z_{kl})=\sum_{n=1}^{d-1} z_{ij}^n z_{kl}^n.$ We denote $M\equiv [1,\dots,d]$ and $I=\{(i,j): i<j, j=2,\dots,N\}$\newline

The diagonal components $a_i$ and $b_j^m$ become the variables that approach infinity along $K$. In turns, the non-diagonal components, $z_{ij}^m$ are bounded, see \eqref{eqn 76} and \eqref{eqn 81}.
Set
\begin{equation} \label{eq_for_rho}
    \begin{aligned}
    \rho_{ij}^2=&| a_i-a_j|^2+\| b_i-b_j\|^2+\| z_{ij}\|^2\\
    \widehat{\rho}_{ij}^2=&| a_i-a_j|^2+\| b_i-b_j\|^2\\\rho^2=&\sum_{i=1}^{N-1}\sum_{j>i}^N\rho_{ij}^2
    \end{aligned}
\end{equation}
where each pair $i,j$ identifies an $u(2)$ sector of the $su(N)$ matrix with components $ib_i,ib_j$ on the diagonal and $z_{ij}$ as the upper non-diagonal component. All $z_{ij}^m$ are independent but $\sum_{i=1}^N a_i=0$ and $\sum_{i=1}^N b_i^m=0$
for each $m=1,\dots,d-1.$ 

Set $\widetilde{C}>0$ arbitrarily large. In the above notation the valley $K$ can be written as the union of two further sets. That is $K=K_{-}\cup K_{+}$ where $K_{-}$ is bounded and $K_{+}$ is unbounded. Explicitly,
\begin{equation}
K_{-}=\{X\in K:\rho_{ij}\le \widetilde{C}\} \qquad \text{and} \qquad
K_+=\{X\in K:\rho_{ij}> \widetilde{C}\}.
\end{equation}
We work first on $K_{+}$. According to the gauge fixing procedure described in Appendix~\ref{C}, $\rho_{ij}>\widetilde{C}$ implies $\widehat{\rho}_{ij}>\epsilon$ where $\epsilon$ is proportional to $\widetilde{C}$. Whence $\epsilon$ approaches infinity when $\widetilde{C}$ does, see \eqref{eqn 81}. Consequently, on $K_{+}$ we have $\widehat{\rho}_{ij}>\epsilon$. Thus
\begin{equation}
    \frac{\rho_{ij}^2}{\widehat{\rho}_{ij}^2}=1+\frac{\| z_{ij}\|^2}{\widehat{\rho}_{ij}^2}<1+\frac{V_0}{\epsilon^4}
\end{equation}
in that part of the valley.

Consider now one component of $z_{ij}^m$. If $z_{ij}^m=C_{ij}^m+iD_{ij}^m$ we denote $C_{ij}$ the real part of one component. 
We then have 
$(C_{ij})^2<\frac{V_0}{\widehat{\rho}_{ij}^2}<\frac{V_0}{\epsilon^2}$. Hence, given that all other components of the matrices $X$ are fixed, $C_{ij}$ takes values on the interval $-\gamma_{+}<C_{ij}<\gamma_{+}$. Below we write $\gamma_{-}=-\gamma_{+}$, so
\begin{equation} \label{estimategammas}
\gamma_+-\gamma_-=2\gamma_+<\frac{2V_0^{1/2}}{\widehat{\rho}_{ij}}<\frac{2V_0^{1/2}}{\epsilon}.    
\end{equation}

In the statements below we consider the wave functions $\Psi(p)$ for $p\in K$, originally in the Fock space $\mathcal{F}(K)$ with components $\Psi_{I}\in H_0^1(K)$. Here and elsewhere, 
\begin{equation}
\overline{\Psi}\cdot \Psi=\sum_I \overline{\Psi}_I\Psi_I 
\qquad \text{and} \qquad \nabla \overline{\Psi}\cdot \nabla \Psi=\sum_I \nabla \overline{\Psi}_I\cdot \nabla \Psi_I .\end{equation}

\begin{lemma} \label{lemma1}
Let $\epsilon>0$ and $V_0>0$ be fixed. There exist two positive constant $C_1,C_2>0$, such that 
\begin{equation} \label{conclusionlemma1}
    \int_K
\rho\overline{\Psi}\cdot \Psi\le C_1\int_K \overline{\Psi}\cdot \Psi+
C_2 \int_K \nabla\overline{\Psi}\cdot \nabla \Psi
\end{equation}
for all $\Psi\in H^1_0(K)$.  Moreover, $C_2$ can be chosen arbitrarily small whenever $\epsilon$ is sufficiently large. 
\end{lemma}
\begin{proof}
Since $C^\infty_{\mathrm{c}}(K)$ is a dense subspace of $H^1_0(K)$ and $\rho$ is smooth, it suffices to find a constant independent of the wave function such that  the inequality holds true for all $\Psi\in C^\infty_{\mathrm{c}}(K)$. 

The argument in the rest of the proof follows Poincar{\'e}'s Lemma, using in addition that the transverse components decrease, as $\frac{V_0^{1/2}}{\widehat{\rho}_{ij}}$ along the valleys extend to infinite. 
 
At $\gamma_{-}$, $\Psi_I=0$, because $\Psi_I\in C^\infty_{\mathrm{c}}(\widetilde{K})$. Then, whenever $p\in K_+$, each component $\Psi_I(p)$ can be written as 
 \begin{equation*}
    (\rho_{ij}\overline{\Psi}_I\Psi_I)^{1/2}=\int_{\gamma_{-}}^{C_{ij}} \mathrm{d}\gamma \, \left[\partial_{\gamma}(\rho_{ij}\overline{\Psi}_I\Psi_I)^{1/2}\right].
\end{equation*}
Applying the Cauchy-Schwartz Inequality, we get
 \begin{equation*}
    (\rho_{ij}\overline{\Psi}_I\Psi_I)^{1/2}\le (\gamma_+-\gamma_{-})^{1/2}\left\{\int_{\gamma_{-}}^{\gamma_{+}} \mathrm{d}\gamma\, \, \left[\partial_{\gamma}(\rho_{ij}\overline{\Psi}_I\Psi_I)^{1/2}\right]^2 \right\}^{1/2}.
\end{equation*}
 Consequently, as $\frac{\partial \rho_{ij}}{\partial \gamma}=\rho_{ij}^{-1}\gamma$ and because of \eqref{estimategammas},
 \begin{equation*}
    \rho_{ij}\overline{\Psi}_I\Psi_I\le \frac{4V_0^{1/2}}{\widehat{\rho}_{ij}}\int_{\gamma_{-}}^{\gamma_{+}} \mathrm{d}\gamma\, \left\{ \frac{1}{4}\gamma^2\rho^{-3}_{ij}\overline{\Psi}_{I}\Psi_I+\rho_{ij}[\partial_{\gamma}(\overline{\Psi}_I\Psi_I)^{1/2}]^2\right\}.
\end{equation*}
Moreover,
$\gamma^2\rho^{-3}_{ij}\le \widehat{\rho}_{ij}^{-1}(\gamma^2\rho^{-2}_{ij})\le \widehat{\rho_{ij}}^{-1}$ hence, after a straightforward calculation,
\begin{equation*}
 \rho_{ij}\overline{\Psi}_I\Psi_I\le\frac{V_0^{1/2}}{\epsilon^2}\int_{\gamma_{-}}^{\gamma_{+}} \mathrm{d}\gamma\,  \overline{\Psi}_{I}\Psi_I+4V_0^{1/2}\left[1+\frac{V_0^{1/2}}{\epsilon^2}\right]\int_{\gamma_{-}}^{\gamma_{+}} \mathrm{d}\gamma\, [\partial_{\gamma}(\overline{\Psi}_I\Psi_I)^{1/2}]^2.
\end{equation*}

Considering all derivatives of $(\overline{\Psi}_I\Psi_I)^{1/2}$ we then get
\begin{equation*}
    \rho_{ij}\overline{\Psi}_I\Psi_I\le\frac{V_0^{1/2}}{\epsilon^2}\int_{\gamma_{-}}^{\gamma_{+}} \mathrm{d}\gamma\,  \overline{\Psi}_{I}\Psi_I+4V_0^{1/2}\left[1+\frac{V_0^{1/2}}{\epsilon^2}\right]\int_{\gamma_{-}}^{\gamma_{+}} \mathrm{d}\gamma\, (\nabla\overline{\Psi}_I\cdot \nabla\Psi_I)
\end{equation*}
 which is valid for each $i,j$ - $su(2)$ sector. We have used the inequality
 \begin{equation}\nabla|\Psi_I|\cdot\nabla|\Psi_I|\le \nabla\overline{\Psi}_I\cdot\nabla\Psi_I\end{equation}
 with $|\Psi_I|=(\overline{\Psi}_I\Psi_I)^{1/2}$, which is always valid. We may now integrate on  $K_{+}$, to get
 \begin{equation}\label{eqn 5.7}
     \int_{K_+}\rho_{ij}\overline{\Psi}_I\Psi_I\le\frac{2V_0}{\epsilon^3} \int_{K_+}\overline{\Psi}_I\Psi_I+\frac{8V_0}{\epsilon}\left[1+\frac{V_0^{1/2}}{\epsilon^2}\right]\int_{K_+}\nabla\overline{\Psi}_I\cdot \nabla\Psi_I
 \end{equation}
 
On the other hand, on $K_{-}$, we have 
 \begin{equation}\label{eqn 5.8}
     \int_{K_-}\rho_{ij}\overline{\Psi}_I\Psi_I\le\widetilde{C} \int_{K_-}\overline{\Psi}_I\Psi_I.
 \end{equation}
 Since \eqref{eqn 5.7} and \eqref{eqn 5.8} are valid for each component $\Psi_I$,
 \begin{equation}\label{eqn 5.9}
     \int_{K}\rho_{ij}\overline{\Psi}_I\cdot\Psi_I\le C_1 \int_{K}\overline{\Psi}\cdot \Psi + C_2\int_{K}\nabla\overline{\Psi}\cdot \nabla \Psi
 \end{equation}
where
 \begin{equation*}
     C_1=\max\left(\frac{2V_0}{\epsilon^{3}},\widetilde{C}\right), \quad C_2= \frac{8V_0}{\epsilon}\left(1+\frac{V_0^{1/2}}{\epsilon^2}\right)
 \end{equation*}
  For the final claim, it is enough to take $\epsilon\to \infty$ to get $C_2\to 0$. After summation on all the $u(2)$ sectors the constants $C_1$ and $C_2$ acquire a factor $\frac{N(N-1)}{2}$.
\end{proof}

The fact that we can ensure $C_2$ becomes arbitrarily small is crucial when proving the coercivity part of the next result.  Note that $C_1$ could be very large and this has no significant effect in the statement.

\begin{lemma} \label{lemma2}
The Hamiltonian of the Supermembrane, valued on the $su(N)$ algebra, has a Dirichlet form coercive on the Fock space $\mathcal{F}(K)$.
\end{lemma}

\begin{proof}
Firstly, the bosonic potential can be expressed in terms of the matrices $X^m$, valued on the $su(N)$ algebra, where $m=1,\dots,d,$
\begin{equation} X^m=\sum_A X^m_AT_A, \quad \operatorname{Tr}  T_A T_{\mathrm{B}}=\delta_{AB}.\end{equation} 
The diagonal components, associated to the diagonal generators can be re-written, in terms $(a_i-a_j)$ and $(b_k^m-b_l^m)$ components with a bounded Jacobian. The non-diagonal components $z_{ij}^m$, coincide with the corresponding components of the non-diagonal generators $T_A$.

Combining the linearity of $V_F$ in $X$ with Lemma~\ref{lemma1}, it follows that there exists a constant $C>0$ such that
 \begin{equation}\int_{K}\overline{\Psi}\cdot V_F\Psi\ge -C\int_{K}\rho\overline{\Psi}\cdot \Psi\ge-C C_1 \int_{K}\overline{\Psi}\cdot \Psi-C C_2\int \nabla\overline{\Psi}\cdot\nabla\Psi.\end{equation}
 
Given $V_0$, choose $\epsilon$ large enough, such that $CC_2<1$.
Then
\begin{align*}
\int_K \nabla \overline{\Psi}\cdot \nabla \Psi&+\int_K \overline{\Psi}\cdot V_\mathrm{B} \Psi +\int_K \overline{\Psi}\cdot V_{\mathrm{F}} \Psi \geq \\
&(1-C C_2)\int_K    \nabla \overline{\Psi}\cdot \nabla \Psi+
 \int_K \overline{\Psi}\cdot (V_\mathrm{B}-C C_1) \Psi.
\end{align*}
Consequently, for the Hamiltonian operator $H=-\Delta+V_{\mathrm{B}}+V_F$,
 \begin{equation}\label{eqn 5.14}(\Psi,H\Psi)_{L^2(K)}\ge \lambda(\Psi,\Psi)_{L^2(K)}+\widehat{C}\| \Psi\|_{H^1(K)}\end{equation}
 for all $\Psi\in\mathcal{F}(K)$. Here $\lambda$ and $\widehat{C}$ are real constants, satisfying $\widehat{C}>0$ and $\lambda$ is possibly negative but it is bounded from below. This implies that the Dirichlet form of the Hamiltonian is coercive on $\mathcal{F}(K)$ as claimed.
\end{proof}


\section{Main results} \label{5}

Let 
\begin{equation}
     D(\Psi,\Phi)=\int_K \overline{\nabla \Psi} \cdot \nabla \Phi +
\int_K \overline{\Psi}\cdot (V_{\mathrm{B}}+V_{\mathrm{F}})\Phi 
\end{equation}
be the Dirichlet form associated to the left hand side of \eqref{eqn 5.14}. Let
\begin{equation}
     \tilde{D}(\Psi,\Phi)=D(\Psi,\Phi)-\lambda (\Psi,\Phi)_{L^2(K)}
\end{equation}
where the parameter $\lambda$ is as in the proof of Lemma~\ref{lemma2}.
Since
\begin{equation}
\begin{aligned}
     |(\Psi&,V_{\mathrm{F}}\Phi)|_{L^2(K)} \leq \\ & \frac{1}{2}|(\Psi-\Phi,V_{\mathrm{F}}(\Psi-\Phi))_{L^2(K)}|+ \frac{1}{2}|(\Psi-i\Phi,V_{\mathrm{F}}(\Psi-i\Phi))_{L^2(K)}|\\&+|(\Psi,V_{\mathrm{F}}\Psi)_{L^2(K)}|+|(\Phi,V_{\mathrm{F}}\Phi)_{L^2(K)}|,
\end{aligned}
\end{equation}
it is readily seen that $\tilde{D}(\Psi,\Phi)$ is strongly coercive and bounded in $H^1_0(K)$.  Then, by virtue of the Lax-Milgram Theorem, it follows that there exists a bounded operator $T:L^2(K)\longrightarrow H^1_0(K)\cap H^2_{\mathrm{loc}}(K)$ such that for any $\Xi\in L^2(K)$ 
\begin{equation}
      \tilde{D}(\Psi,T\Xi)=(\Psi,\Xi)_{L^2(K)} \qquad \text{for all }\Psi\in H^1_0(K).
\end{equation}
That is $T=(H-\lambda)^{-1}$ is a well defined bounded operator. This fact has an important consequence highlighted below. 
   
We know that $H^1(K)$ is compactly embedded in $L^2(K)$,  [Theorem~3, in Section~3 of \cite{berger-schechter}]. Hence, the composition of $T$ with the inclusion operator from $H^1_0(K)\cap H^2_{\mathrm{loc}}(K)\subset H^1(K)$ into $L^2(K)$, is a compact operator on $L^2(\Omega)$. In other words, the resolvent $T$ of $H$ at $\lambda$, is compact. This implies two main consequences.

\begin{lemma} The Hamiltonian $H$ considered in this paper, with domain in  $H^1_0(K)$, has a purely discrete spectrum of eigenvalues, each of finite multiplicity, with no accumulation point other than $+\infty$.
\end{lemma}
\begin{proof} This is a direct consequence of the coercivity and the compact embedding, as the resolvent of $H$ becomes a compact operator.\end{proof} 

\begin{theorem}
Given $\Xi\in L^2(K)$, there exists a unique $\Psi\in H^1_0(K)$ such that 
\begin{equation}  \label{theorem1}
     D(\Phi,\Psi)=(\Phi,\Xi)_{L^2(K)} \qquad \text{for all } \Phi\in H^1_0(K).
\end{equation}
\end{theorem}
\begin{proof}
Since the spectrum of $H$ comprises only isolated eigenvalues of finite multiplicity, we just have to verify that $\operatorname{Ker}(H)=\{0\}$. If $H\widetilde{\Psi}=0$ for some non-zero $\widetilde{\Psi}\in H^1_0(K)$, then $D(\widetilde{\Psi},\widetilde{\Psi})=0$. Hence $Q(\widetilde{\Psi})=Q^{\dagger}(\widetilde{\Psi})=0$.
Thus, from Lemma~1 of \cite{PLB2019}, $\widetilde{\Psi}=0$.
This is clearly a contradiction. So indeed $\operatorname{Ker}(H)=\{0\}$ and by the Fredholm Alternative, the stated result follows.
\end{proof}

Note that the regularity properties of elliptic operators ensure that in the above theorem, $\Psi\in H^1_0(K)\cap H^2_{\mathrm{loc}}(K)$. Consequently $\Psi$ is the unique solution of
\begin{equation}
    H \Psi=\Xi \ \text{ in }K \qquad  \Psi\in H^1_0(K)\cap H^2_{\mathrm{loc}}(K).
\end{equation}
Given $\Theta\in H^2(K)$, set $\Xi=-H\Theta\in L^2(K)$ and $\Phi=\Psi+\Theta\in H^1_0(K)\cap H^2_{\mathrm{loc}}$. Then 
\begin{align*}
  H\Phi&=0 \qquad \text{in }K \\
  \Phi&=\Theta \qquad \text{on }\partial K.
\end{align*}
That is, $\Phi$ is the unique solution to the homogeneous Dirichlet problem  associated with the region $K$ for the Hamiltonian $H$. 

If in the above equations we impose the constraint $\varphi^A\Psi=0$ for all $A=1,\ldots,(N^2-1)$ and $\Xi\in H^1(K)$, then 
\begin{equation}
\varphi^A \Xi=\varphi^A H \Psi=H\varphi^A\Psi=0.
\end{equation}
That is, $\Xi$ also satisfies the constraint. Hence, in the search of the ground state, there is no loss of generality when imposing the constraint on $\Xi$. For the potential we have, we know that
\begin{equation}
    \varphi^A V_{\mathrm{B}}(X,z,\overline{z})=0
\end{equation}
for all indices $A$ also. And for all real parameters $\xi_A$
\begin{equation}
    \xi_A\varphi^A X_i^C=\xi_A f^{ABC} X_{iB}\quad \text{and} \quad
\xi_A\varphi^A z^C=\xi_A f^{ABC} z_{B}.
\end{equation}
Consider in $\mathbb{R}^n$ the transformation
\begin{equation}
    X^C_i\mapsto X_i^C+\xi_A\varphi^A X_i^C,\quad
    z^C\mapsto z^C+\xi_A\varphi^A z^C,\quad 
    \overline{z}^C\mapsto \overline{z}^C+\xi_A\varphi^A \overline{z}^C
\end{equation} 
where $\xi$ is an infinitesimal parameter. Then
\begin{align*}
    &V_{\mathrm{B}}(X_i^C+\xi_A\varphi^A X_i^C,z^C+\xi_A\varphi^A z^C,\overline{z}^C+\xi_A\varphi^A \overline{z}^C)= \\
&V(X,z,\overline{z})+\xi_Af^{ABC} \left(X_{iB} \frac{\partial}{\partial X^C_i}+z_{\mathrm{B}} \frac{\partial}{\partial z^C} +\overline{z}_{\mathrm{B}} \frac{\partial}{\partial z^C} \right) V_{\mathrm{B}}(X,z,\overline{z}) \\ &+
O(|\xi|^2).
\end{align*}
Thus, the constraint generates transformations of coordinates on $\mathbb{R}^n$ which preserve the value of $V_{\mathrm{B}}$. Furthermore $\Psi=0$ on $\partial K$. Hence
\begin{equation}
    \Psi(X+\xi\varphi X,z+\xi \varphi z,\overline{z}+\xi \varphi \overline{z},\theta+\xi \varphi \theta)=0
\end{equation}
on $\partial K$, because the fields on the $\theta$ expansion are also evaluated at $\partial K$ and by the conditions of equation they are zero. 
Since
\begin{align*}
\Psi(X+\xi\varphi X,z+\xi \varphi z,&\overline{z}+\xi \varphi \overline{z},\theta+\xi \varphi \theta)|_{\partial K}= \\
& \Psi(X,z,\overline{z},\theta)|_{\partial K} +
\xi \varphi \Psi(X,z,\overline{z},\theta)|_{\partial K}+O(|\xi|^2),
\end{align*}
we then obtain 
\begin{equation}
    \varphi^A\Psi(X,z,\overline{z},\theta)=0 \qquad \text{on }\partial K.
\end{equation}
Moreover, if $\Psi\in H^1_0(K)\cap H^2_{\mathrm{loc}}(K)$, we get
$\varphi^A\Psi\in H^1_0(K)$. All this ensures the validity of the following.

\begin{lemma}
Let $\Theta\in H^3(\mathbb{R}^n)$ and $\varphi^A\Theta=0$ for $A=1,\ldots,(N^2-1)$. Then the solution $\Psi$ in the context of the weak problem \eqref{theorem1} also satisfies $\varphi^A\Psi=0$.
\end{lemma}
\begin{proof}
 Set $\Xi=-H\Theta\in H^1(K)$ as above. Then 
\begin{equation}
    \varphi^A\Xi=H(\varphi^A \Psi)=0 
\end{equation}
with $\varphi^A\Psi\in H^1_0(K)\cap C^{\infty}(K)$. Hence,
for each index $A$, 
\begin{equation}
     Q\varphi^A\Psi=0 \quad \text{and} \quad Q^{\dagger}\varphi^A\Psi=0
\end{equation}
Thus, according to \cite[Lemma~1]{PLB2019}, we conclude that $\varphi^A\Psi=0$ in $K$.
\end{proof}

from this lemma we gather that the homogeneous problem
\begin{equation}\label{final}
\begin{aligned}
   H\Phi&=0 \qquad \text{in }K \\
   \varphi^A\Phi&=0 \qquad \text{in } K \\
\Phi&=\Theta \qquad \text{on } \partial K
\end{aligned}
\end{equation}
has a unique solution $\Phi\in H^1(K)\cap H^2_{\mathrm{loc}}(K)$
for any given $\Theta\in H^3(K)$ satisfying the constraint
$\varphi^A \Theta=0$. As we have already noticed there is no loss of generality by imposing the constraint on $\Theta$.


\section{Conclusions}  \label{6}
In this work we fully examine the Hamiltonian of the regularized $SU(N)$  Supermembrane in eleven dimensions on an unbounded region. The region is naturally connected with the theory, as it is defined by the set $K=\{ X^{mA}: V_{\mathrm{B}}(X)<V_0\}$. These are the so-called valleys of the bosonic potential. Importantly, on these valleys, a) there are sub-regions extending to infinity where $V_{\mathrm{B}}$ vanishes, b) the potential is dominated by the fermionic sector and c) the full potential is unbounded from below. 

It is well known that, despite of a), the bosonic Hamiltonian defined on the unrestricted space has discrete spectrum with finite multiplicity. This is in contrast to the well-known fact, shown in \cite{dwln}, that  the spectrum of the supersymmetric Hamiltonian defined on the unrestricted space is continuous and comprises the whole segment $[0,\infty)$. One of our main contributions presently is the fact that, remarkably, the supersymmetric Hamiltonian restricted to the valleys and for wavefunctions vanishing on the boundary, has discrete spectrum with finite multiplicity. Notably, and in agreement with the established result, the wavefunctions constructed in \cite{dwln} for the proof of continuity of the spectrum do not (and must not) vanish on this boundary. Moreover, since these regions are preserved by the action of the $SU(N)$ constraint, the formulation of the present $SU(N)$ regularization restricted to the valleys, is both natural and well posed. 

Our findings suggest several puzzling avenues of further enquiry. Do the eigenvalues survive as embedded modes inside the continuous spectrum for the unrestricted space? In the context of (Super) Yang Mills theory, if  they survive and the slow mode regime captures relevant aspects of confinement, can these eigenvalues describe glueball boundstates and flux tubes connecting quarks when the theory is properly compactified  to $D=4$? Equivalently, in the regularized $SU(N)$ description of the Supermembrane compactified to four dimensions, do these eigenvalues capture aspects of QCD confinement?

We also establish the existence and uniqueness of the state which is annihilated by the Hamiltonian on the physical subspace determined  by the $SU(N)$ constraint, satisfying a prescribed boundary condition on $\partial K$. The proof of this fact has three main ingredients. i) The volume of $K$ is finite, subject to constraints on $d$, the number of transverse directions in the light cone gauge. ii) The fermionic potential satisfies a crucial estimate (see Section~\ref{4}), which renders a coercive Hamiltonian. iii) The embedding $H^1(K)\subset L^2(K)$ is compact according to know results in \cite{PLB2019}. All this is in agreement with the previous findings of \cite{Lundholm, austing1,austing2}. 

Although we consider explicitly the supersymmetric Hamiltonian of the $SU(N)$ regularized  $D=11$ Supermembrane, the estimates we found for the fermionic potential rely only on the linear dependence of the bosonic coordinates. We therefore expect that the present findings could be extrapolated to the $SU(N)$ Supermembrane on the other admissible dimensions, provided the restrictions of Lemma~\ref{lemfinitevolK} are fulfilled.  The $SU(N)$ Matrix Models with $N\ge 4$ spacetime and dimensions $D\ge 5$ contain a massless ground state, consequently our proof does not include the dimension $D=4$. However in the context of Super Yang Mills theories, for  $SU(N)$ gauge  groups with  $N=2,3$, the spacetime dimensions must be with $D\ge 7,6$ respectively. If (Super) Yang Mills theories in the slow regime are a good indicator of the fully-fledged theory confinement behaviour, these results could suggest the need of extra dimensions. Our claimed results are valid for a Supermembrane Theory  (large $N$) formulated in $D=5, 7,11$ only.

By Domain Monotonicity, many of our claims extend to a formulation of the theory on any reasonably regular region inside the valleys (with suitable boundary and $SU(N)$ constraints). For instance star-shaped domains for large enough $V_0$. Therefore an asymptotic analysis of the groundstate of the regularized $SU(N)$ supermembrane is perhaps possible, by considering a sequence of Dirichlet problems on regions taking $V_0\to \infty.$ 
We have shown the existence and uniqueness of the solution of the homogeneous Dirichlet problem and hence the existence and uniqueness  of the state annhiliated by the supersymmetric Hamiltonian. However, this is generically not annihilated by the supersymmetric charges, as it is only for a particular boundary condition that the corresponding state may be annihilated. 

The unique state that we have determined by solving the Dirichlet problem, is the minimizer of the norm defined in terms of the supersymmetric charges, namely $\vert\vert\varphi\vert\vert^2_Q\equiv (Q\varphi,Q\varphi)_{L^2(K) }+(Q^{\dag}\varphi, Q^{\dag}\varphi)_{L^2(K) },$
 for a given boundary condition on $\partial K$.  Perhaps it would be possible to pursue further studies of the (Super) Yang Mills Theory in the slow mode regime, when it is confined to this tubular/star-shaped region.

\section*{Acknowledgements}  AR and MPGM were partially supported by Projects Fondecyt 1161192 (Chile). AR was partially supported by MINEDUC-UA project code ANT1855. 

\appendix

\section{Measure of the \texorpdfstring{$su(2)$}{}  sectors in \texorpdfstring{$su(N)$}{}} \label{A}
In this first appendix we analyze the Lebesgue measure of diagonal matrices that belong to the Cartan subalgebra of the $su(N)$ algebra (describing the longitudinal directions  along the valleys), when we consider Matrix Models for different rank of the $su(N)$ gauge groups.

The potential $V_{\mathrm{B}}(X)$ is invariant under conjugation by $U\in su(N)$,
\begin{equation}\label{eqn II}
X^m\to UX^mU^{-1}.
    \end{equation}

We can diagonalize one of the $X^m$ matrices, say $\widehat{X}$. If the eigenvalues of $\widehat{X}$ are all different, then $U$ is determined up to a diagonal matrix acting on the right to $U$. If there are equal eigenvalues, then $U$ has additionally non-diagonal terms undetermined and we can fix some of the components of another matrix $X$. This procedure will be explained in detail in due course.

We characterize the measure $\widehat{X}$ of the  matrices of the Cartan subalgebra with different rank, in order to obtain inductively the  expression for arbitrary $N$.
Associated to the $su(2)$ matrix $\widehat{X}$
\begin{equation}\label{eqn 0}
 \widehat{X}\equiv 
    \begin{pmatrix}
   ia & m\\
   -\overline{m}& -ia
    \end{pmatrix}
\end{equation}
there is the following measure $\mathcal{M}_2(\widehat{X})=\rho^2\ \mathrm{d}\rho \, \mathrm{d}\Omega$, with $\rho^2=a^2+m\overline{m}$.  We can fix $m=0$ by a selection of $U$. Hence the measure can be expressed as $\mathcal{M}_2(\widehat{X})=a^2\ \mathrm{d}\vert a\vert \mathrm{d}\Omega.$\newline
 For a $su(3)$ matrix, there are three $su(2)$ sectors
 
 \begin{equation}\begin{aligned}\label{eqn 1}
   & \begin{pmatrix}
   \frac{i}{3}\alpha & m_{12}& 0\\
   -\overline{m}_{12}& -\frac{i}{3}\alpha &0\\ 0&0&0
    \end{pmatrix}+
    \begin{pmatrix}
   -\frac{i}{3}\beta & 0& m_{13}\\
  0&0&0\\ -\overline{m}_{13}&0&\frac{i}{3}\beta
    \end{pmatrix}+ \begin{pmatrix}
  0&0&0\\
  0&-\frac{i}{3}\gamma & m_{23}\\ 0&-\overline{m}_{23}&\frac{i}{3}\gamma
    \end{pmatrix}=\\ &
     \begin{pmatrix}
  \frac{i}{3}(\alpha-\beta)& m_{12}& m_{13}\\
  -\overline{m}_{12}& \frac{-i}{3}(\alpha+\gamma)& m_{23}\\ -\overline{m}_{13}&-\overline{m}_{23}&\frac{i}{3}(\beta+\gamma)
    \end{pmatrix}
     \end{aligned}
\end{equation}
 We may perform a linear change of variables from the original $\widehat{X}^A$ coordinates to the new ones $\alpha, \beta,\gamma, m_{ij}$. Each $su(2)$ sector contributes to the measure as in \eqref{eqn 0}, we then have the following measure associated to $\widehat{X}$,
 \begin{equation}\mathcal{M}_3(\widehat{X})=\alpha^2\beta^2\gamma^2 \ \mathrm{d}\vert\alpha\vert \, \mathrm{d}\vert\beta\vert \, \mathrm{d}\vert\gamma\vert \delta(\gamma-\alpha-\beta).\end{equation}
 
 We can shift $\alpha$ and $\beta$ by the same amount $\lambda$ while $\gamma$ by $-\lambda$ and \eqref{eqn 1} remains invariant. We can always fix this invariance by taking $\gamma-\alpha-\beta=0$, in agreement with the total number of degrees of freedom of a $su(3)$ matrix.
The measure can be expressed in terms of the diagonal components of the $su(3)$
 matrix, $a_1\equiv \frac{1}{3}(\alpha-\beta)$, $a_2\equiv -\frac{1}{3}(\alpha+\gamma)=-\frac{1}{3}(2\alpha+\beta)$, and $a_3\equiv \frac{1}{3}(\beta+\gamma).$ In fact, by defining $a_{ij}= a_i-a_j$ \begin{equation}a_{12}=\alpha,\quad a_{13}=\beta,\quad a_{32}=\alpha+\beta.\end{equation}
 The measure can be then re-expressed as 
 \begin{equation}\mathcal{M}_3(\widehat{X})=\vert a_{12}\vert^2\vert a_{13}\vert^2\vert a_{23}\vert^2 \mathrm{d}\vert a_{12}\vert \, \mathrm{d}\vert a_{13}\vert.\end{equation}

In general for $su(N)$ there are $\frac{N(N-1)}{2}$ $u(2)$ sectors. Each sector is defined by two diagonal components, say $b_i,b_j$ and the non-diagonal component $z_{ij}$. For example, for $su(4)$ there are six sectors and the measure can be expressed as
\begin{equation}\mathcal{M}_4(\widehat{X})=\vert a_{12}\vert^2\vert a_{13}\vert^2\vert\vert a_{14}\vert^2 \vert a_{23}\vert^2\vert a_{24}\vert^2\vert a_{34}\vert^2 \mathrm{d}\vert a_{12}\vert\, \mathrm{d}\vert a_{13}\vert\, \mathrm{d}\vert a_{14}\vert.\end{equation}
The main point to write the measure in this form, is that all the factors in the bracket will be cancelled from a contribution during the integration procedure in the evaluation of $\mathrm{Vol}(K)$. The above expression can be generalized to $su(N)$ in a straightforward way.


\section{Bounds on the measure  \texorpdfstring{$\mathrm{Vol}(K)$}{}} \label{B}
We now consider $\mathrm{Vol}(K)$ in the region where all the eigenvalues of $\widehat{X}$ are different. We denote by $\mathcal{N}$  the region where  all the differences satisfy $\vert a_{ij}\vert >\epsilon$ for all $i\ne j$. We find the bounds for three cases: $su(2)$, $su(3)$, $su(4)$ and finally the general expression for $su(N)$.

\subsection{The \texorpdfstring{$su(2)$}{} case} \label{B.1} In this case $a_2=-a_1$, then $a_1-a_2=2a>\epsilon$. The expression of the potential is 
\begin{equation}\label{eqn 2}
    \frac{1}{2}V_{\mathrm{B}}= 4(a^2+b^2)\vert\vert z\vert\vert^2-4\vert b\cdot z\vert^2 +\vert\vert z\vert\vert^4-(z\cdot z)(\overline{z}\cdot \overline{z}).
\end{equation}
It follows from \eqref{eqn 2} that $\vert z\vert^2<\frac{V_0}{\epsilon^2}$.
Define
\begin{equation}
    \widetilde{V}_{\mathrm{B}}= 8(a^2+\vert b\vert^2)\vert\vert z\vert\vert^2-\vert b\cdot z\vert^2.
\end{equation}
Since $ \vert\vert z\vert\vert^4-(z\cdot z)(\overline{z}\cdot \overline{z})=4[\operatorname{Re}(z)^2(\operatorname{Im}(z)^2- (\operatorname{Re}(z)\cdot \operatorname{Im}(z))^2]\ge 0$ we have 
\begin{equation}\frac{1}{2}\widetilde{V}_{\mathrm{B}}(X) \le \frac{1}{2}V_{\mathrm{B}}(X)\end{equation}
for all $a,b,z.$

We decompose $z=\lambda b+z^{\perp}$, where $b\cdot z^{\perp}=0.$ Then 
\begin{equation}\frac{1}{8}\widetilde{V}_{\mathrm{B}}=a^2\vert \lambda\vert^2+(a^2+\vert\vert b\vert\vert^2)\vert\vert z^{\perp}\vert\vert^2.\end{equation}

The set $\widetilde{K}\equiv \{ X: \widetilde{V}_{\mathrm{B}}(X)<V_0\}$, for $a$ and $b$ fixed, is the interior of an ellipsoid $E$ described by the coordinates $\{\operatorname{Re}(\lambda),\operatorname{Im}(\lambda), \operatorname{Re}(z^{\perp}), \operatorname{Im}(z^{\perp})\}.$ Consequently 
\begin{equation}
\begin{aligned} \label{eqn 7}
\mathrm{Vol}(\widetilde{K}\cap \mathcal{N})=& \widetilde{C}\int_{\widetilde{K}\cap \mathcal{N}} \mathrm{d} a (d\vert\vert b\vert\vert)  \left[a^2\vert\vert b\vert\vert^{d-2} (a^2(a^2+\vert\vert b\vert\vert ^2)^{d-2})^{-1}\right]\\&
\le C \int_{\widetilde{K}\cap \mathcal{N}} \mathrm{d}\rho \rho^{d-1} (\rho^{2(d-2)})^{-1}
\end{aligned}
\end{equation}
where $\rho^2=a^2+\vert\vert b\vert\vert^2$ and $\mathrm{d}a \ \mathrm{d}\vert\vert b\vert= \rho \ \mathrm{d}\rho\, \mathrm{d}\varphi$ and we used $\vert\vert b\vert\vert^{d-2}\le\rho^{d-2}.$  The factor $(a^2(a^2+\vert\vert b\vert\vert ^2)^{d-2})^{-1}$ corresponds to the volume of $E$. In $\widetilde{K}\cap \mathcal{N}$, $\rho^2>\frac{\epsilon^2}{4}$
 since $2\vert a\vert>\epsilon$. Consequently, the above integral is convergent, provided
 \begin{equation}\label{eqn 20}
     2(d-2)-(d-1)>1.
 \end{equation}
We therefore conclude that $\mathrm{Vol}(\widetilde{K}\cap \mathcal{N})$ is finite, and hence $\mathrm{Vol}(K\cap \mathcal{N})$ is also finite for $d\ge 5$.

In the context of Supermembrane theory, $d$ parametrizes transverse dimensions. Hence the measure of the valleys of the $su(2)$ bosonic potential are finite for $D\ge 7$. Recall that Supermembrane Theory is consistently defined in $4,5,7$ and $11$ dimensions.  The factor $a^2$ in the measure of the integral is cancelled by a factor $a^{-2}$ arising from the volume  of $E$, associated to the coordinate $\vert\lambda\vert$. This cancellation occurs also for the $su(N)$ potential, as we will see later on.
 
 \subsection{The \texorpdfstring{$su(3)$}{}  case}  \label{B.2}
The explicit expression for the potential is 
 \begin{equation}
 \begin{aligned}\label{eqn 3}
     V_{\mathrm{B}}=&\frac{1}{2}\sum_{m,n}(1,1)^{mn}\overline{(1,1)}^{mn}+(2,2)^{mn}\overline{(2,2)}^{mn}+(3,3)^{mn}\overline{(3,3)}^{mn}\\&+\sum_{m,n}(1,3)^{mn}\overline{(1,3)}^{mn}+(2,3)^{mn}\overline{(2,3)}^{mn}+(1,2)^{mn}\overline{(1,2)}^{mn}
     \end{aligned}
 \end{equation}
where $(i,j)^{mn}$ denotes the $i,j$ component of the matrix $[X^m,X^n]$. The diagonal components $(i,i)^{mn}, i=1,2,3$, depend solely on the non-diagonal components of $X^m$ and $X^n$. The contribution of the matrix $\widehat{X}$ to \eqref{eqn 3} is only to the non-diagonal terms and it is 
\begin{equation}2\vert a_{12}\vert^2 \vert\vert z_{12}\vert\vert^2+2\vert a_{13}\vert^2 \vert\vert z_{13}\vert\vert^2+ 2\vert a_{23}\vert^2\vert\vert z_{23}\vert\vert^2.\end{equation}
There is no contribution of $\widehat{X}$ to the diagonal terms $(i,i)$ in the expression of the potential since $\widehat{X}$ is diagonal. 
In the set $K\cap \mathcal{N}$, $\vert\vert z_{ij}\vert\vert$ are bounded, in fact \eqref{eqn 3} implies
\begin{equation}\label{22}
\vert\vert z_{ij}\vert\vert^2<\frac{V_0}{\epsilon^2}, \quad i,j=1,2,3.
\end{equation}

In the case of the $su(2)$ algebra the sum of diagonal terms in \eqref{eqn 3} correspond to the term in the last two terms on the right in \eqref{eqn 2}.

The non-diagonal terms of $V_{\mathrm{B}}$ can be re-arranged in terms of the three $u(2)$ sectors 
\begin{equation}V_{12}+V_{13}+V_{23},\end{equation} where
\begin{equation}
\begin{aligned}
V_{ij}=&\vert a_{ij}\vert^2\vert\vert z_{ij}\vert\vert^2+ \vert\vert b_{ij}\vert\vert^2 \vert\vert z_{ij}\vert\vert^2-\vert (b_{ij})\cdot z_{ij}\vert^2\\& -i[(b_{ij})\cdot \overline{z}_{ik}](z_{ij}\cdot z_{jk})+i[(b_{ij})\cdot {z}_{jk}](z_{ij}\cdot \overline{z}_{ik})\\& +i[(b_{ij})\cdot {z}_{ik}](\overline{z}_{ij}\cdot \overline{z}_{jk})-i[(b_{ij})\cdot \overline{z}_{jk}](\overline{z}_{ij}\cdot z_{ik})\\ & +\vert\vert z_{ik}\vert\vert^2\vert\vert z_{jk}\vert\vert^2-\vert z_{ik}\cdot z_{jk}\vert^2.
\end{aligned}
\end{equation}
for $b_{ij}\equiv b_i-b_j$, and  $i\ne j\ne k$ and $i,j,k =1,2,3$. We define 
\begin{equation}
\begin{aligned}\label{eqn 8}
\widetilde{V}_{\mathrm{B}}\equiv & V_{12}+V_{13}+V_{23}- (\vert\vert z_{13}\vert\vert^2\vert\vert z_{23}\vert\vert^2-\vert z_{13}\cdot {z}_{23}\vert^2)\\&
(\vert\vert z_{12}\vert\vert^2\vert\vert z_{23}\vert\vert^2-\vert z_{12}\cdot \overline{z}_{23}\vert^2)-(\vert\vert z_{12}\vert\vert^2\vert\vert z_{13}\vert\vert^2-\vert z_{12}\cdot {z}_{13}\vert^2).
\end{aligned}
\end{equation}
Then 
 \begin{equation}
     \widetilde{V}_{\mathrm{B}}<V_{\mathrm{B}}
 \end{equation}
for all $a_i,b_j,z_{kl}$ since the terms with brackets in \eqref{eqn 8} are positive.
We then define 
\begin{equation}\widetilde{K}=\{ a_i,b_j,z_{kl}: \widetilde{V}_{\mathrm{B}}<V_0\}.\end{equation}
 On each $u(2)$ sector we can shift the corresponding $z_{ij}$, in order to simplify the expression for $V_{ij}$. 

 For the three $u(2)$ sectors,
 \begin{equation}
     z_{ij}=\lambda_i\frac{b_{ij}}{\vert\vert b_{ij}\vert\vert}+z_{ij}^{\perp}, \quad b_{ij}\cdot z_{ij}^{\perp}=0.
 \end{equation}
Define
 \begin{equation}
     \widetilde{z}_{ij}=z_{ij}^{\perp}+\frac{i}{\rho_{ij}^2}\{[b_{ij}\cdot z_{ik}]\overline{z}_{jk}-[b_{ij}\cdot \overline{z}_{jk}]{z}_{ik}\}
 \end{equation}
where $\rho^2_{ij}=\vert a_{ij}\vert^2+\vert\vert b_{ij}\vert\vert^2.$ 
Then, \begin{equation}
 \begin{aligned}
 \widetilde{V}_{ij}=& V_{ij}-(\vert\vert z_{ik})\vert\vert^2\vert\vert z_{jk}\vert\vert^2-\vert z_{ik}\cdot z_{jk}\vert^2)\\ & =\vert a_{ij}\vert^2\vert \lambda_i\vert^2+\rho_{ij}^2
\vert\vert \widetilde{z}_{ij}^{\perp}\vert\vert^2-A_{ij}^2 \end{aligned}
 \end{equation}
 where \begin{equation}A_{ij}^2= \frac{1}{\rho^2_{ij}}\vert\vert (b_{ij}\cdot z_{ik})\overline{z}_{jk}-(b_{ij}\cdot \overline{z}_{jk})z_{ik}\vert\vert^2\end{equation}
 is bounded from above
 \begin{equation}
     A_{ij}^2<\left(\frac{2V_0}{\epsilon^2}\right)^2.
 \end{equation}
 Consequently, $A_{ij}^2+A_{ik}^2+A_{jk}^2$ is bounded by $12\left(\frac{V_0}{\epsilon^2}\right)^2.$

The measure of $\widetilde{K}\cap \mathcal{N}$ is bounded by the measure of the set $\widehat{K}\cap \mathcal{N}$, with 
\begin{equation}\label{eqn 27b} \widehat{K}= \left\{a,b,\lambda,\widetilde{z}: \widehat{V}_{\mathrm{B}}<V_0+12\left(\frac{V_0}{\epsilon^2}\right)^2\right\}\end{equation}
where
\begin{equation}
\begin{aligned}
    \widehat{V}_{\mathrm{B}} &\equiv \vert a_{12}\vert^2\lambda_1^2+ \vert a_{13}\vert^2\hat{\lambda}_1^2+ \vert a_{23}\vert^2\hat{\hat{\lambda}}_1^2+\\
    &+\rho_{12}^2\vert\vert \widetilde{z}_{12}^{\perp}\vert\vert^2+
    +\rho_{13}^2\vert\vert \widetilde{z}_{13}^{\perp}\vert\vert^2+
    +\rho_{23}^2\vert\vert \widetilde{z}_{23}^{\perp}\vert\vert^2.
    \end{aligned}
\end{equation}

For $su(N)$, the factor $12$ in \eqref{eqn 27b} changes to $2N(N-1)$.
As in the $su(2)$ case, given $a_{ij}$, $\rho_{ij}\equiv \vert a_{ij}\vert^2+\vert\vert b_{ij}\vert\vert^2$, the set of points satisfying 
$\widehat{V}_{\mathrm{B}}<\widehat{V}_0= V_0+12\left(\frac{V_0}{\epsilon^2}\right)^2$ coincides with the interior of an ellipsoid $E$ determined by radii  \begin{equation}
\widehat{V}_0\left\{\frac{1}{\vert a_{12}\vert}, \frac{1}{\vert a_{13}\vert},\frac{1}{\vert a_{23}\vert},\frac{1}{\rho_{12}},\frac{1}{\rho_{13}},\frac{1}{\rho_{23}}\right\}.\end{equation} 
Its volume is then
\begin{equation}\label{eqn III}
    \mathrm{Vol}(E)=\Pi_{i<j}\frac{C}{\vert a_{ij}\vert^2\cdot \rho_{ij}^{2(d-2)}}
\end{equation}
 We notice that in $\widehat{K}\cap \mathcal{N}$ each factor in the denominator of \eqref{eqn III} is bounded away from zero since $\vert a_{ij}\vert >\epsilon, \quad i,j=1,2,3$ and $i\ne j$.
 
 Using the measure of $\widehat{X}$, denoted by  $\mathcal{M}_3(\widehat{X})$,  and the product of the measure of $b_{ij}$, denoted as  $\mathcal{M}(b_{ij})$ with $i\ne j$, we get
  \begin{equation}\Pi_{i<j}\mathcal{M}(b_{ij})=\vert\vert b_{12}\vert\vert^{d-2}\vert\vert b_{13}\vert\vert^{d-2}\mathrm{d}\vert\vert b_{12}\vert \vert \mathrm{d}\vert\vert b_{13}\vert\vert \mathrm{d}\Omega_{12}\mathrm{d}\Omega_{13}.\end{equation}
 From \eqref{eqn III} we obtain that the measure of $\widehat{K}\cap \mathcal{N}$ ,  is given by 
 \begin{equation}
    \begin{aligned} \label{eqn IV}
     \mathcal{M}_3&(\widehat{K}\cap \mathcal{N})= \int \mathcal{M}_3(\widehat{X})\Pi_{i\ne j}\mathcal{M}(b_{ij}) \mathrm{Vol} (E)\\ &\le C\int \mathrm{d}\vert a_{12}\vert \mathrm{d}\vert a_{13}\vert \mathrm{d}\vert\vert b_{12}\vert\vert \mathrm{d}\vert\vert b_{13}\vert\vert \mathrm{d}\Omega_{12}\mathrm{d}\Omega_{13}\left(\rho_{12}^{d-2}\rho_{13}^{d-2}\rho_{23}^{2(d-2)}\right)^{-1}.
     \end{aligned}
 \end{equation}
 
  In order to evaluate this integral, consider the $d$-dimensional vectors $r_i\equiv (a_i,b_i)$ with $i=1,2,3$, satisfying $r_1+r_2+r_3=0$. Observe that
  $\rho_{ij}=\vert\vert r_i-r_j\vert\vert$.
  Define
  $u_{12}\equiv r_1-r_2$, and $u_{13}\equiv \frac{r_1-r_3}{\vert\vert r_1-r_2\vert\vert}.$ Then
 \begin{equation}
     \rho_{12}=\vert\vert u_{12}\vert\vert,\quad \rho_{13}=\vert\vert u_{13}\vert\vert \cdot\vert\vert u_{12}\vert\vert,\quad \rho_{23}=\vert\vert u_{13}-\frac{u_{12}}{\vert\vert u_{12}\vert\vert}\vert\vert \cdot\vert\vert u_{12}\vert\vert.
 \end{equation}
The integral \eqref{eqn IV} can then be expressed as an integral in $u_{12}$ and $u_{13}$, 
\begin{equation}
    \begin{aligned} \label{eqn V}
      \mathcal{M}_3(\widehat{K}\cap \mathcal{N})&=
     \int \vert\vert u_{12}\vert\vert^3 d\vert\vert u_{12}\vert\vert d\omega_{12}\vert\vert u_{13}\vert\vert d\vert\vert u_{13}\vert\vert d\omega_{13}\cdot\\ &
     (\vert\vert u_{12}\vert\vert^{4(d-2)}\cdot \vert\vert u_{13}\vert\vert^{d-2}\cdot \vert\vert u_{13}-\frac{u_{12}}{\vert\vert u_{12}\vert\vert}\vert\vert^{2(d-2)})^{-1}.
     \end{aligned}
 \end{equation}
 
 As noted before, since we are integrating on $\mathcal{N}$, each factor in the denominator is bounded away from zero. 
  The third power arises from the following expressions
  \begin{equation}
      \begin{aligned}
     & d\vert a_{12}\vert d \vert\vert b_{12}\vert\vert \to \vert\vert u_{12}\vert\vert d \vert\vert u_{12}\vert\vert d\varphi_{12}\\ &
     d\vert a_{13}\vert d \vert\vert b_{13}\vert\vert\to \rho_{13}d\rho_{13}d\varphi_{13}= \vert\vert u_{12}\vert\vert^2 \vert\vert u_{13}\vert\vert d \vert\vert u_{13}\vert\vert d\varphi_{13}.
      \end{aligned}
  \end{equation}
 The term on the right hand side of \eqref{eqn V} factorizes into two integrals,
 \begin{equation}\label{eqn VI}
  I_1^{(3)}=\int d\vert\vert u_{12}\vert\vert \vert\vert u_{12}\vert\vert^{-4(d-2)+3}   
 \end{equation}
 and 
 \begin{equation}\label{eqn VII}
 I_2^{(3)}= \int d\vert\vert u_{13}\vert\vert \vert\vert u_{13}\vert\vert^{-(d-2)+1} \vert\vert u_{13}- \frac{u_{12}}{\vert\vert u_{12}\vert\vert}\vert\vert^{-2(d-2)} d\omega_{12}\omega_{13}. 
 \end{equation}
 So, in order to have a convergent integral, each factor must be finite, and we then require
 \begin{equation}\label{eqn VIII}
     4(d-2)-3>1 \qquad \Rightarrow \qquad d>3\Rightarrow \qquad d\ge 4,
 \end{equation}
and
  \begin{equation}\label{eqn IX}
     3(d-2)-1>1\qquad \Rightarrow\qquad d>2+\frac{2}{3}\qquad\Rightarrow\qquad d\ge 3,
 \end{equation}
from equations \eqref{eqn VI} and \eqref{eqn VII}, respectively. If $d\ge 4$, then the volume of the valley for the $su(3)$ algebra is finite. The restriction arising from the integration on $\vert\vert u_{12}\vert\vert$ in \eqref{eqn VIII} is stronger than \eqref{eqn IX}, because of the factor $\vert\vert u_{12}\vert\vert^3$. This also occurs for the $su(N)$ case.
\subsection{The \texorpdfstring{$su(4)$}{} case} \label{A4}
Following the same procedure as above, we obtain the same bound \eqref{22} for all the non-diagonal components of the matrices $X^m$. We end up with the integral 
\begin{equation} \begin{aligned}\label{eqn 27bb}
\mathcal{M}_4(\widehat{K}\cap \mathcal{N})=\int & d\vert a_{12}\vert d\vert a_{13}\vert d\vert a_{14}\vert d\vert \vert b_ {12}\vert\vert d\vert \vert b_{13}\vert\vert d\vert \vert b_{14}\vert\vert
d\Omega_{12} d\Omega_{13} d\Omega_{14}\mathbb{R}_4\\
&\mathbb{R}_4=(\rho_{12}^{d-2}\rho_{13}^{d-2}\rho_{14}^{d-2}\rho_{23}^{2(d-2)}\rho_{24}^{2(d-2)}\rho_{34}^{2(d-2)})^{-1}
\end{aligned}\end{equation}

Set, as before, the variables $r_i$ for $i=1,\dots,4$ satisfying  $\sum_{i=1}^4 r_i=0$. Then
\begin{equation}\rho_{ij}=\vert\vert r_i-r_j\vert\vert=\vert\vert (r_i-r_1)+(r_1-r_j)\vert\vert.\end{equation}
For
\begin{equation} u_{12}=r_1-r_2, \quad u_{13}=\frac{r_1-r_3}{\vert\vert r_1-r_2\vert\vert} \quad \text{and} \quad u_{14}=\frac{r_1-r_4}{\vert\vert r_1-r_2\vert\vert}\end{equation}
it follows that, 
\begin{equation}
\begin{gathered}
\rho_{12}=\vert\vert u_{12}\vert\vert, \quad\rho_{13}=\vert\vert u_{13}\vert\vert\cdot \vert\vert u_{12}\vert\vert,\quad \rho_{14}=\vert\vert u_{14}\vert\vert\cdot \vert\vert u_{12}\vert\vert,\\
\rho_{23}=\vert\vert u_{13}-\frac{u_{12}}{\vert\vert u_{12}\vert\vert}\vert\vert \cdot\vert\vert u_{12}\vert\vert,\quad \rho_{24}=\vert\vert u_{13}-\frac{u_{12}}{\vert\vert u_{12}\vert\vert}\vert\vert \cdot\vert\vert u_{12}\vert\vert, \\
 \rho_{34}=\vert\vert u_{14}- u_{13} \vert\vert \cdot\vert\vert u_{12}\vert\vert.
\end{gathered}\end{equation}

Although the integral \eqref{eqn 27b} can be performed without using the following bound
\begin{equation} \rho_{34}=\vert\vert r_3-r_4\vert\vert>\epsilon,\quad \rho^{-2(d-2)}_{34}<\epsilon^{-2(d-2)},\quad\textrm{for} \quad d>2,\end{equation}
this not change the restriction on integral dimensions $d.$ We may then dismiss the factor $\rho^{-2(d-2)}_{34}$ since \eqref{eqn 27b} is bounded by an integral which factorizes into $\mathcal{M}_4(\widehat{K}\cap \mathcal{N})=I_1^4\cdot I_2^4$ for
\begin{equation} \label{eqn XI}
I_1^4=\int d\vert\vert u_{12}\vert\vert \vert\vert u_{12}\vert\vert^5 (\vert\vert u_{12}\vert\vert^{7[d-2]})^{-1}
\end{equation}
and
\begin{equation} \label{eqn XII} I_2^4=\epsilon^{-2(d-2)}\int  d\vert\vert u_{13}\vert\vert d\vert\vert u_{14}\vert\vert  d\Omega_{12} d\Omega_{13} d\Omega_{14}\mathbb{A}\end{equation}
with \begin{align*}\mathbb{A}=\vert\vert u_{13}&\vert\vert^{1-(d-2)} \vert\vert \vert\vert u_{14}\vert\vert^{1-2(d-2)}\\&\vert\vert u_{13}-\frac{u_{12}}{\vert\vert u_{12}\vert\vert}\vert\vert^{-2(d-2)} \vert\vert u_{14}-\frac{u_{12}}{\vert\vert u_{12}\vert\vert}\vert\vert^{-2(d-2)}.\end{align*}
If \begin{equation} \label{eqn XIII} 7(d-2)-5>1\quad\Rightarrow\quad d>2+\frac{6}{7}\quad\Rightarrow\quad d\ge 3 \end{equation}
and \begin{equation} \label{eqn XIV}3(d-2)-1>1\quad\Rightarrow\quad d>2+\frac{2}{3}\quad \Rightarrow\quad d\ge 3,\end{equation}
the integral \eqref{eqn 27b} is convergent. If $d\ge 3$ the measure of the valley for the $su(4)$ algebra is finite. There are no divergences arising for factors going to zero, since we are working in the region $\mathcal{N}$.

\subsection{Bounds in the \texorpdfstring{$su(N)$}{} case} \label{B4}  We obtain the same bound \eqref{22} for all the non-diagonal components of the matrices $X^m$. The result, concerning the finiteness of the measure, follows directly by dismissing all the terms involving $\rho_{ij}$ for all $i\ge 3$ and all $j>i$, because all of them are bounded by powers of $\epsilon$. The integral representing the measure of $K\cap\mathcal{N}$ is then bounded by an integral which factorizes into two integrals. An integral on $\vert\vert u_{12}\vert\vert$, with positive powers $(N-1)+(N-2)=2N-3$ arising from 
 \begin{equation}d \vert a_{1i}\vert d\vert\vert b_{1i}\vert\vert\to \rho_{1i}d\rho_{1i}d\varphi_i\end{equation}
 with 
\begin{equation}
      \rho_{1i}\to\begin{cases}
      \vert\vert u_{12}\vert\vert & i=2 \\
      \vert\vert u_{1i}\vert\vert \cdot \vert\vert u_{12}\vert\vert  & i>2.
      \end{cases}
  \end{equation}
This contributes with $(N-1)$ to the exponent of $\vert\vert u_{12}\vert\vert$. Followed by 
\begin{equation}d\rho_{1i}=\vert\vert u_{12}\vert\vert d\vert\vert u_{1i}\vert\vert\end{equation} for $i>2$, which contributes with $(N-2)$ to the exponent. Finally, $(N-1)(d-2)$ arising from the measure factors $\vert\vert b_1-b_i\vert\vert^{d-2}$. The contribution to the negative powers arises from the integrals on $z_{ij}$, $i<j$ and the further change of variables $\rho_{ij}\to \vert\vert u_{ij}\vert\vert$. Since we are only considering the pairs,
$12,\dots, 1N,23,\dots,2N$ factors we have a power $-[(N-1)+(N-2)]2(d-2)$. 

The convergence of the integral is ensured for 
\begin{equation}[(N-1)+(N-2)]2(d-2)-(N-1)(d-2)-(N-1)-(N-2)>1.\end{equation}
That is
 \begin{equation}d>2+\frac{2(N-1)}{3(N-1)-2}.\end{equation}
The term
\begin{equation}\frac{2(N-1)}{3(N-1)-2}= \begin{cases}
2& N=2\\
1& N=3\\
<1& N\ge 4.
\end{cases}
\end{equation}
We thus recover the previous results for $su(2)$ in \eqref{eqn 20}, $su(3)$ in \eqref{eqn VIII}, $su(4)$ in \eqref{eqn XIII} and obtain the general result. \emph{The integral is finite if $d\ge 3$, for $su(N)$ whenever $N\ge 4$.}

The second integral associated to the measure of $\mathcal{M}_N(\widehat{K}\cap \mathcal{N})$ is bounded by the integral
\begin{equation}
\begin{aligned}
   I= &C\int d\vert\vert u_{13}\vert\vert\dots d\vert\vert u_{1N}\vert\vert d\Omega_{12}\dots d\Omega_{1N} \mathbb{B}\quad\textrm{and,}\\
   &\mathbb{B}= \Pi_{i=3}^{N}\vert\vert u_{1i}\vert\vert^{1-(d-2)}\vert\vert u_{1i}- \frac{u_{12}}{\vert\vert u_{12\vert\vert}}\vert\vert^{-2(d-2)}
\end{aligned}
\end{equation}
which is convergent provided 
\begin{equation}3(d-2)-1>1\to d\ge 3.\end{equation}
This occurs for $N=3$ and $N=4$. We then conclude that, if $d\ge 3$, the measure of $\mathcal{M}_N(\widehat{K}\cap \mathcal{N})$ is finite for the algebra $su(N)$ whenever $N\ge 4$. From the viewpoint of the Supermembrane Theory taking $N$ to infinity, the restriction  $D\ge 5$ in the spacetime dimension ensures that the volume of the region defined in terms of the  bosonic potential  along the  valleys is finite. 


\section{Gauge transformations}  \label{C}
Given $X\in u(N)$, we consider the gauge transformations
\begin{equation}\label{A1}
    X^{'}=U^{-1}XU,\quad U\in SU(N).
\end{equation}
Under \eqref{A1}, $\operatorname{Tr} X$ and $\operatorname{Tr} X^{\dag}X$ remain invariant. In particular if $X\in su(N)$, using the notation introduced in Section~\ref{3},
\begin{equation}\label{A2}
N \operatorname{Tr} X^{\dag}X=N \operatorname{Tr} X^{\dag}X-(\operatorname{Tr} X)(\operatorname{Tr}X^{\dag})=\sum_I[(b_{ij})^2+2N\vert Z_{ij}\vert^2],
\end{equation}
is also invariant under \eqref{A1}.
In the case under consideration we have $m=1,\dots, d$ matrices $X^m\in su(N)$. For each of them \eqref{A2} remains invariant under \eqref{A1}. 

We consider the unbounded region, $\sum_m \operatorname{Tr} X^{m\dag }X^m>C^2$. This region decomposes into subsets where at least for one $m$, say $\widetilde{m}$, 
\begin{equation}\label{A3}
 \operatorname{Tr} X^{\widetilde{m}\dag }X^{\widetilde{m}}>\frac{C^2}{d}.
\end{equation}
Hence, from \eqref{A2}, 
\begin{equation}\label{A4}
\sum_I[(b^{\widetilde{m}}_{ij})^2+2N\vert Z_{ij}^{\widetilde{m}}\vert^2]>\frac{NC^2}{d}.
\end{equation}
We perform now a gauge transformation which diagonalize $X^{\widetilde{m}}$. Then, after \eqref{A1}, we have 
\begin{equation}\label{A5}
\sum_I(b^{'\widetilde{m}}_{ij})^2>\frac{NC^2}{d}.
\end{equation}
 
In order to simplify the notation, from here on we do not use the prime for the new components.
Since there are $\frac{N(N-1)}{2}(i,j)$ sectors and taking into consideration \eqref{A5}, for at least one sector $(i,j)$, say $(1,N)$, we must have, 
 \begin{equation}\label{A6}
(b^{\widetilde{m}}_{1N})^2>\frac{2C^2}{d(N-1)}.
\end{equation}
 Note that \eqref{A6} implies
 \begin{equation}\label{A7}
\vert b^{\widetilde{m}}_{1i}\vert+\vert b^{\widetilde{m}}_{iN}\vert \ge \vert b^{\widetilde{m}}_{1N}\vert > \left[\frac{2C^2}{d(N-1)}\right]^{\frac{1}{2}},\quad i=2,\dots,N-1.
\end{equation}
Therefore at least $(N-1)$ pairs $(i,j)\in I$ satisfy $(b^{\widetilde{m}}_{ij})^2>\frac{C^2}{2d(N-1)}.$


\section{The measure of \texorpdfstring{$K$}{}}\label{D}

In the previous appendices we considered the measure of the set $K\cap \mathcal{N}$,
\begin{equation}K\cap\mathcal{N}\equiv \{x\in k: \vert a_i-a_j\vert >\epsilon\quad \textrm{for all}\quad i,j,i<j\}.\end{equation}
Now we invoke these results and show that also the Lebesgue measure of $K$ is finite.

\subsection{Case \texorpdfstring{$su(2)$}{}}

Consider $X^m=\begin{pmatrix} ib^m& z^m\\-\overline{z}^m& -ib^m\end{pmatrix}$, for $m=1,\dots,d.$ We do not distinguish here $\widehat{X}$ from the other $u(2)$ matrices.  Under the gauge transformation \eqref{eqn II}, the traces $\operatorname{Tr} X^mX^{\dag m}$ are invariant for each $m$. Hence,
\begin{equation}P_2^m\equiv \frac{1}{2}\operatorname{Tr} X^mX^{m\dag}=(b^m)^2+z^m\overline{z}^m\end{equation}
is invariant.

We decompose $K$ into a finite number of subsets, determined by whether $P_2^m$ satisfies the condition $P_2^m\le \epsilon^2$ or the condition $P_2^m>\epsilon^2$, for $m=1,\dots,d.$ The subset $P_2^m\le \epsilon^2$ for all $m=1,\dots,d$ has finite measure, so we are left with other subsets for which at least for one $m$, say $\widetilde{m}$, $P_2^{\widehat{m}}>\epsilon.$ We have
\begin{equation} (b^{\widetilde{m}})^2+z^{\widetilde{m}}\overline{z}^{\widetilde{m}}>\epsilon^2.\end{equation}

We now perform a gauge transformation such that $X^{\widetilde{m}}$ becomes diagonal. The new $b^{\widetilde{m}}$ which we denote it with the same letter satisfies 
\begin{equation}(b^{\widetilde{m}})^2>\epsilon^2.\end{equation}
The two eigenvalues $b^{\widetilde{m}}$ and $-b^{\widetilde{m}}$ are such that \begin{equation}\vert b^{\widetilde{m}}- (-)b^{\widetilde{m}}\vert=2\vert b^{\widetilde{m}}\vert>2\vert\epsilon\vert.\end{equation} We may then apply the argument in Appendix~\ref{B.1}.  We conclude that if $d\ge 5$, $\mathrm{Vol}(K)$ is finite. Here $b^{\widetilde{m}}$ plays the role of the component $a$ in the notation of Appendix~\ref{B}.

\subsection{\bf Case \texorpdfstring{$su(3)$}{}}
We consider now the case $N=3$. Following appendix \ref{C}, there are at least two pairs, say $(1,3)$ and $(2,3)$ which satisfy
\begin{equation}\label{A8}
 (b^{\widetilde{m}}_{13})^2 > \frac{C^2}{4d}\quad \text{and} \quad (b^{\widetilde{m}}_{23})^2 > \frac{C^2}{4d}.
\end{equation}
Using \eqref{eqn 3}, we obtain $\vert\vert z_{13}\vert\vert^2<\frac{4dV_0}{C^2}$ and $\vert\vert z_{23}\vert\vert^2<\frac{4dV_0}{C^2}$.
The other pair may satisfy the same inequality or not. In the first case
\begin{equation}\label{A9}
 (b^{\widetilde{m}}_{ij})^2 > \frac{C^2}{4d} \quad \text{for all} \quad (i,j)\in I. 
\end{equation}
In the second case 
\begin{equation}\label{A10}
 (b^{\widetilde{m}}_{12})^2 \le \frac{C^2}{4d}. 
\end{equation}
In both cases $z_{ij}^{\widetilde{m}}=0$, for all $(i,j)\in I$. In the second case we may have 
\begin{equation}\label{A11}
(b^{m}_{12})^2+\vert z_{12}^{m}\vert^2\le\frac{C^2}{4d}
\end{equation}
 for all $m\in M$. That is, all the $(1,2)$ sector is bounded, the $(1,2)$ sector has then finite measure, or for some $\widehat{m}$
 \begin{equation}\label{A12}
h^2=(b^{\widehat{m}}_{12})^2+\vert z_{12}^{\widehat{m}}\vert^2>\frac{C^2}{4d}
\end{equation}
 where $h\ge 0$. 

We now perform a gauge transformation
 \begin{equation}
     \label{A13}U=\begin{pmatrix} A& C & 0\\-\overline{C} & +\overline{A} & 0\\ 0 &0 & 1 \end{pmatrix}.
 \end{equation}
Here $A\overline{A}+C\overline{C}=1$, hence $\det(U)=1$.
Then \begin{equation}\label{A14}
     b_{12}^{'\widehat{m}}=b_{12}^{\widehat{m}}\cdot(A\overline{A}-C\overline{C})+2iz_{12}^{\widehat{m}}\overline{C}\overline{A}-2i\overline{z}_{12}^{\widehat{m}}CA.
 \end{equation} 
 We choose $C=\overline{A}u$, with  $u=\frac{-i(b_{12}^{\widehat{m}})\mp h}{2\overline{z}_{12}^{\widehat{m}}}$, which yields
  \begin{equation}
     \label{A15}
     b_{12}^{'\widehat{m}}=h,\quad z_{12}^{'\widehat{m}}=0.
 \end{equation}
 Under this gauge transformation
  \begin{equation} \label{A16}
    \begin{pmatrix} z_{13}^{'\widetilde{m}}\\z_{23}^{'\widetilde{m}} \end{pmatrix}=  \begin{pmatrix} A& C \\-\overline{C} & -\overline{A} \end{pmatrix}\begin{pmatrix} z_{13}^{\widetilde{m}}\\z_{23}^{\widetilde{m}} \end{pmatrix},
 \end{equation}
 hence $z_{13}^{'\widetilde{m}}=z_{23}^{'\widetilde{m}}=0.$ We then have, from \eqref{A2} and \eqref{A4}, 
 \begin{equation}\label{A17}
(b^{'\widetilde{m}}_{12})^2+\vert z_{12}^{'\widetilde{m}}\vert^2+(b^{'\widetilde{m}}_{13})^2 +(b^{'\widetilde{m}}_{23})^2> \frac{3C^2}{d}.
\end{equation}
But 
 \begin{equation}\label{A18}
(b^{'\widetilde{m}}_{12})^2+\vert z_{12}^{'\widetilde{m}}\vert^2=(b_{12}^{\widetilde{m}})^2\le\frac{C^2}{4d}
\end{equation}
is invariant under the gauge transformation generated by \eqref{A13}. We thus have, from \eqref{A17} and \eqref{A18},
 \begin{equation}\label{A19}
(b^{'\widetilde{m}}_{13})^2+(b_{23}^{'\widetilde{m}})^2>\frac{11C^2}{4d}.
\end{equation}
This implies (from the argument in \eqref{A6}, \eqref{A7} and \eqref{A18}) that
 \begin{equation}\label{A20}
(b^{'\widetilde{m}}_{13})^2>\frac{C^2}{4d},\quad (b_{23}^{'\widetilde{m}})^2>\frac{C^2}{4d}
\end{equation}
and from \eqref{A8} and the  matrix $U$ we are considering \begin{equation}\label{73b}\vert\vert z_{13}^{'}\vert\vert^2+\vert\vert z_{23}^{'}\vert\vert^2<\frac{8dV_0}{C^2}\to \vert\vert z_{12}^{'}\vert\vert^2<\frac{V_0}{\epsilon^2},\end{equation}
for large enough $\epsilon$, proportional to $C$.
Consequently, in $su(3)$, the unbounded region $\sum_m \mathrm{Tr} X^{m\dag}X^m>C^2$ is the union of subsets. In each one of them there exists two unbounded sectors $(i,j)$, which by a gauge transformation satisfy $(b^{\widetilde{m}}_{ij})^2>\frac{C^2}{4d}, z_{ij}^{
\widetilde{m}}=0, (i,j)\in I$ for some $\widetilde{m}\in M$. The third $(k,l)$ sector either satisfies 
\begin{equation}\label{A21}
[(b_{kl}^m)^2+\vert z_{kl}^m\vert^2]\le \frac{C^2}{4d}\end{equation} for all $m\in M$ (\emph{i.e.}  it is bounded) or there exists $\widehat{m}$ such that 
\begin{equation}\label{A22}
(b_{kl}^{\widehat{m}})^2> \frac{C^2}{4d},\quad z_{kl}^{\widehat{m}}=0.\end{equation}
In all cases, we have \begin{equation}\label {eqn 76}\vert\vert z_{ij}\vert\vert^2<\frac{V_0}{\epsilon^2}\end{equation} for all $(i,j)\in I$, and large enough $\epsilon^2$.

We now consider the expression of the potential $V_{\mathrm{B}}$. We do not distinguish any diagonal matrix as in Appendix~\ref{B.2}. The quadratic terms on $b_{ij}^{\widetilde{m}}$ corresponding to the $u(2)$ sector $(i,j)$, are 
\begin{equation}
v_{ij}=\vert\vert b_{ij}\vert\vert^2\vert\vert z_{ij}\vert\vert^2-(b_{ij}\cdot z_{ij})^2.
\end{equation}
More explicitly we have
\begin{equation}\label{eqn 7.2.2}
v_{ij}=\vert b_{ij}^{\widetilde{m}}\vert^2 \vert\vert z_{ij}\vert\vert^2+\sum_{M\setminus\widetilde{m}}\vert b_{ij}^{m}\vert^2 \sum_{M\setminus\widetilde{m}}\vert z_{ij}^{n}\vert^2-\sum_{M\setminus\widetilde{m}}\vert b_{ij}^{m} z_{ij}^{m}\vert^2
\end{equation}
and we obtain an analogue expression to the ones in Appendix~\ref{B.2}. Although there are linear terms on the diagonal components in the potential, the bound \eqref{eqn 76} allows to show that $\vert\vert z_{ij}\vert\vert^2<\frac{V_0}{\vert b_{ij}^{\widetilde{m}}\vert^2}$, the whole argument of finite volume follows identical steps. When a $u(2)$ sector is bounded, it is always possible to eliminate it from the expression of $V_{\mathrm{B}}$ and the calculations for the unbounded sectors, restricted by the bounded one, are as in Appendix~\ref{B.2}. The measure of $K$ for the $su(3)$ algebra is finite subject to the same conditions as before. 


\subsection{Case \texorpdfstring{$su(N)$}{}}
Let  $Y_N^m\in u(N), m\in M$. Under a gauge transformation 
\begin{equation}\label{eqn 7.3.1}
    Y_N^m\to U^{-1}Y_N^mU,\quad m\in M,\quad U\in su(N),
\end{equation}
the traces $\mathrm{Tr} Y_M^m$ and $\mathrm{Tr} Y^{m\dag}_NY^m_N$ remain invariant. Consequently,
\begin{equation}\label{eqn 7.3.2}
   P^m\equiv \sum_I(\vert b_{ij}^m\vert^2+2N\vert z_{ij}^m\vert^2), \quad m\in M,
\end{equation}
remains also invariant under \eqref{eqn 7.3.1}.

We decompose $K$ into a finite number of subsets. the subset for which 
\begin{equation}\sum_m P^m\le C^2\end{equation}
for arbitrary $C>0$, has all variables $b_{ij}^m,z_{ij}^m$ bounded. On the complement, $\sum_m P^m> C^2$, there always exists $\widetilde{m}$ for which $P^{\widetilde{m}}>\frac{C^2}{9}$. That is, in all subsets of the complement at least for one index the condition is satisfied. $Y^{\widetilde{m}}_N$ can be diagonalized by a gauge transformation \eqref{eqn 7.3.1}. From \eqref{eqn 7.3.2} we obtain 
\begin{equation}\label{eqn 7.3.3}
   \sum_I\vert b_{ij}^{\widetilde{m}}\vert^2>\frac{C^2}{9}, \quad  z_{ij}^{\widetilde{m}}=0.
\end{equation}

There are at least $(N-1)$ $ b_{ij}^{\widetilde{m}}$
satisfying
\begin{equation}\label{eqn 7.3.4}
  \vert b_{ij}^{\widetilde{m}}\vert^2>\frac{C^2}{36h}, \quad  h=\frac{1}{2}N(N-1).
\end{equation} Then
\begin{equation} \vert\vert z_{ij}\vert\vert^2<\frac{36hV_0}{C^2}.
\end{equation} 
In fact, at least one, say $b_{12}^{\widetilde{m}}$, must satisfy 
$\vert b_{12}^{\widetilde{m}}\vert^2>\frac{C^2}{9h}.$
Then 
\begin{equation}\label{eqn 7.3.5}
 \vert b_{1i}^{\widetilde{m}}\vert +\vert b_{2i}^{\widetilde{m}}\vert>\vert b_{12}^{\widetilde{m}}\vert
 >\left(\frac{C^2}{9h}\right)^{1/2}, \quad  \textrm{for all}\quad  i=3,\dots N.
\end{equation}
This implies \eqref{eqn 7.3.4}. 

Set \begin{equation}J\equiv \{(i,j): i<j \quad 
\textrm{and}\quad b_{ij}^{\widetilde{m}} \text{ satisfy \eqref{eqn 7.3.4}}\},\end{equation} which contains $N-1$ pairs. 
In order to simplify the argumentation, consider the case $j=N$\begin{equation}J=\{(i,N):i=1,\dots,N-1\},\end{equation} so its complement becomes \begin{equation}\overline{J}=\{(i,j):i<j,\quad i,j\in [1,\dots,N-1]\}.\end{equation} In this case, 
\begin{equation}U=\begin{pmatrix} U_1& 0\\ 0& 1 \end{pmatrix}, \quad U_1\in su(N-1).\end{equation}
Under this gauge transformation $b_{NN}^{\widetilde{m}}$ and $\sum_i^{N-1}\vert z_{iN}^{\widetilde{m}}\vert^2$ remain invariant for each $m$. Hence, defining $Y_{N-1}^m\in u(N-1)$, where we have excluded the $N^{th}$ row and column of $Y_N^m$, the traces $\operatorname{Tr} Y_{N-1}^{m}$ and $\operatorname{Tr} Y_{N-1}^{m\dag}Y_{N-1}^m$ are invariant and so is 
\begin{equation}\label{eqn 7.3.5}
 Q^m= \sum_{\overline{J}}(\vert b_{ij}^m\vert^2 +2(N-1)\vert z_{ij}^m\vert^2). 
\end{equation}

Moreover, the action of $U_1$ preserves the condition \eqref{eqn 7.3.4}, which in this case is 
\begin{equation}\label{eqn 7.3.7a}
  \vert b_{iN}^{\widetilde{m}}\vert^2>\frac{C^2}{36h},\quad  z_{iN}^{\widetilde{m}}=0, \quad \vert\vert z_{iN}\vert\vert^2<\frac{36hV_0}{C^2}.
\end{equation}
We have then reduced the $u(N)$ case to the $u(N-1)$ case, satisfying \eqref{eqn 7.3.7a}. Furthermore, since in the case of $u(3)$ we have shown that each $u(2)$ sector is either bounded or there exists an index $\widetilde{m}$  for each sector $(i,j)$ satisfying \eqref{eqn 76}, we conclude that the statement of the following lemma is valid.
\begin{lemma}
For $m\in M$, let $Y^m\in u(N)$. Let $\mathcal{C}>0$ be  constant. Then, each $u(2)$ sector, $b_{ij}^m,z_{ij}^m$, is either bounded or there exists an index $\widetilde{m}$, depending on the sector, for which
\begin{equation}\label{eqn 7.3.7}
  \vert b_{ij}^{\widetilde{m}}\vert>\mathcal{C},\quad  z_{ij}^{\widetilde{m}}=0.
\end{equation}
\end{lemma}

We now consider the Lebesgue measure of $K$ for the potential, valued on a $su(N)$ algebra. The bounded $u(2)$ sectors are dismissed from the potential leaving only those which are unbounded. For each one of these, there is only one quadratic term on $b_{ij}^{\widetilde{m}}$ as in \eqref{eqn 7.2.2} in the potential, although there is also linear terms on it, we obtain
\begin{equation}\label{eqn 81}\vert\vert z_{ij}\vert\vert^2<\frac{V_0}{\vert b_{ij}^{\widetilde{m}}\vert^2}<\frac{V_0}{\epsilon^2}\end{equation} for large enough $\epsilon^2$. We then reduce the evaluation of the subsets of $K$ to the particular cases of section B. This is enough to show that $\operatorname {Vol}(K)<\infty$, under
the hypothesis of Lemma~\ref{lemfinitevolK}.


\end{document}